\documentclass{CSML}
\pdfoutput=1

\usepackage{lastpage}
\newcommand{\dOi}{13(4:26)2017}
\lmcsheading{}{1--\pageref{LastPage}}{}{}%
{May~31,~2016}{Dec.~27,~2017}{}

\usepackage{hyperref}
\hypersetup{hidelinks}

\usepackage{amssymb}
\usepackage[english]{babel}
\usepackage{amsmath}
\usepackage{amsfonts}
\usepackage[latin1]{inputenc}
\usepackage{latexsym}
\usepackage{graphics}
\usepackage{epic}
\usepackage{graphicx}

\begin{document}

\newtheorem{defn}{Definition}
\newtheorem{definition}[defn]{Definition}
\newtheorem{exmp}{Example}
\newtheorem{example}[exmp]{Example}
\newcommand{\tuple}[1]{\ensuremath{\langle #1 \rangle}}
\newcommand{\cpt}{\ensuremath{\operatorname{cpt}}}
\newcommand{\ignore}[1]{}

\title[The Power of AC for  Partially-Ordered Forbidden Patterns]{The Power of Arc Consistency for CSPs  
Defined by Partially-Ordered Forbidden Patterns}

\thanks{An extended abstract of this article appeared in Proc. of LICS'16~\cite{cz16:lics}.}

\author[Cooper and \v{Z}ivn\'y]{Martin C. Cooper$^1$}
\address{$^1$IRIT, University of Toulouse III, France}
\email{cooper@irit.fr}
\thanks{$^1$Supported by EPSRC grant EP/L021226/1}

\author[]{Stanislav \v{Z}ivn\'y$^2$}
\address{$^2$Dept. of Computer Science, University of Oxford, UK}
\email{standa.zivny@cs.ox.ac.uk}
\thanks{$^2$Supported by EPSRC grant EP/L021226/1 and a Royal Society University
Research Fellowship. Part of this work was done while the second author was
visiting the Simons Institute for the Theory of Computing at UC Berkeley. This
project has received funding from the European Research Council (ERC) under the
European Union's Horizon 2020 research and innovation programme (grant agreement
No 714532). The paper reflects only the authors' views and not the views of the
ERC or the European Commission. The European Union is not liable for any use
that may be made of the information contained therein} 

\keywords{arc consistency, constraint satisfaction problem, forbidden pattern, tractability}
\subjclass{Logic and constraint programming} 

\begin{abstract}
Characterising tractable fragments of the constraint satisfaction
problem (CSP) is an important challenge in theoretical computer
science and artificial intelligence. Forbidding patterns (generic
sub-instances) provides a means of defining CSP fragments which are
neither exclusively language-based nor exclusively structure-based.
It is known that the class of binary CSP instances in which the
broken-triangle pattern (BTP) does not occur, a class which includes
all tree-structured instances, are decided by arc consistency (AC), a
ubiquitous reduction operation in constraint solvers. We provide a
characterisation of simple partially-ordered forbidden patterns which
have this AC-solvability property. It turns out that BTP is just one
of five such AC-solvable patterns. The four other patterns allow us
to exhibit new tractable classes.
\end{abstract}

\maketitle

\thicklines \setlength{\unitlength}{1pt}
\newsavebox{\varthree}
\savebox{\varthree}(20,60){
\begin{picture}(20,60)(0,0)
\put(10,30){\oval(18,58)} \put(10,10){\makebox(0,0){$\bullet$}}
\put(10,30){\makebox(0,0){$\bullet$}} \put(10,50){\makebox(0,0){$\bullet$}}
\end{picture}
}
\newsavebox{\vartwo}
\savebox{\vartwo}(20,40){
\begin{picture}(20,40)(0,0)
\put(10,20){\oval(18,38)} \put(10,10){\makebox(0,0){$\bullet$}}
\put(10,30){\makebox(0,0){$\bullet$}}
\end{picture}
}
\newsavebox{\varone}
\savebox{\varone}(20,40){
\begin{picture}(20,40)(0,0)
\put(10,20){\oval(18,28)} \put(10,20){\makebox(0,0){$\bullet$}}
\end{picture}
} \setlength{\unitlength}{1pt}

\section{Introduction}

The \emph{constraint satisfaction problem} (CSP) provides a common framework for
many theoretical problems in computer science as well as for many real-life
applications. A CSP instance consists of a number of variables, a domain, and
constraints imposed on the variables with the goal to determine whether the
instance is satisfiable, that is, whether there is an assignment of domain
values to all the variables in such a way that all the constraints are
satisfied.

The general CSP is NP-complete and thus a major research direction is to
identify restrictions on the CSP that render the problem \emph{tractable}, that
is, solvable in polynomial time.

A substantial body of work exists from the past two decades on applications of
universal algebra in the computational complexity of and the applicability of
algorithmic paradigms to CSPs. Moreover, a number of celebrated results have
been obtained through this method; see~\cite{Barto14:survey} for a recent
survey. However, the algebraic approach to CSPs is only applicable to
\emph{language-based} CSPs, that is, classes of CSPs defined by the set of
allowed constraint relations but with arbitrary interactions of the constraint
scopes. For instance, the well-known 2-SAT problem is a class of language-based
CSPs on the Boolean domain $\{0,1\}$ with all constraint relations being
\emph{binary}, that is, of arity at most two.

On the other side of the spectrum are \emph{structure-based} CSPs, that is,
classes of CSPs defined by the allowed interactions of the constraint scopes but
with arbitrary constraint relations.
Here the methods that have been successfully used to establish complete
complexity classifications come from graph
theory~\cite{Grohe07:jacm,Marx13:jacm}.

The complexity of CSPs that are neither language-based nor structure-based, and
thus are often called \emph{hybrid} CSPs, is much less understood;
see~\cite{Carbonnel15:constraints,cz17:survey} for recent surveys.
One approach to hybrid CSPs that has been rather successful studies the
classes of CSPs defined by \emph{forbidden patterns}; that is, by forbidding
certain generic subinstances. The focus of this paper is on such CSPs.
We remark that we deal with \emph{binary} CSPs but, unlike in most papers on (the
algebraic approach to) language-based CSPs, the domain is \emph{not} fixed and
is part of the input.

An example of a pattern is given in Figure~\ref{fig:btpmc}(a) on
page~\pageref{fig:btpmc}. This is the
so-called \emph{broken triangle} pattern (BTP)~\cite{cjs10:aij-btp} (a formal
definition is given in Section~\ref{sec:prelim}). BTP is an example of a \emph{tractable} pattern,
which means that the class of all binary CSP instances in which BTP does not occur is
solvable in polynomial time. The class of CSP instances defined by forbidding
BTP includes, for instance, 
all tree-structured binary
CSPs~\cite{cjs10:aij-btp}. There are several generalisations of BTP, for
instance, to quantified CSPs~\cite{Gao11:aaai}, to existential
patterns~\cite{ccez15:jcss}, to patterns on non-binary constraints~\cite{cooper14:cp-broken}, and
other classes~\cite{Naanaa13:jtai,Cooper15:k-btp}.

The framework of forbidden patterns is general enough to capture language-based
CSPs in terms of their polymorphisms. 
If $\Gamma$ is a finite set of (binary) constraint relations, then CSP($\Gamma$)
is the set of instances whose constraint relations all belong to $\Gamma$. It is
well known that a necessary condition for CSP($\Gamma$) to be tractable
(assuming P $\neq$ NP) is that $\Gamma$ has a non-trivial
polymorphism~\cite{Bulatov05:classifying}. For any given a polymorphism $f$, the
framework of forbidden patterns is general enough to capture the class of
(binary) CSP instances CSP($\Gamma_f$), where $\Gamma_f$ is the set of (binary)
constraint relations closed under $f$, provided patterns can be enriched by the
function $f$ (or an equivalent relation) on domain elements.
For instance, the pattern in Figure~\ref{fig:btpmc}(b)  on
page~\pageref{fig:btpmc} captures the notion of
binary relations that are max-closed~\cite{Jeavons95:maxclosed}. 

Surprisingly, there are essentially only two classes of algorithms (and their
combinations) known for establishing tractability of CSPs. These are, firstly, a
generalisation of Gaussian elimination~\cite{Bulatov06:maltsev,Dalmau06:gen},
whose applicability for language-based CSPs is known~\cite{Idziak10:siam}, and,
secondly, problems solvable by \emph{local consistency methods}, which
originated in artificial intelligence; see references
in~\cite{Rossi06:handbook}. The latter can be defined in many equivalent ways
including pebble games, Datalog, treewidth, and proof
complexity~\cite{Feder98:monotone}.
Intuitively, a class of CSP instances is solvable by $k$-consistency if
unsatisfiable instances can always be refuted while only keeping partial
solutions of size $k$ ``in memory''. For instance, the 2-SAT problem is solvable
by local consistency methods.

For structure-based CSPs, the power of consistency methods is well
understood: a class of structures can be solved by $k$-consistency if and only if the treewidth
(modulo homomorphic equivalence) is at most $k$~\cite{Atserias07:power}.
Consequently, consistency  methods solve all tractable cases of
structurally-restricted bounded-arity CSPs~\cite{Grohe07:jacm}. For
language-restricted CSPs, the power of consistency methods has only recently
been characterised~\cite{Barto14:jacm,Bulatov09:width}.

\subsection*{Contributions}

Our ultimate goal is to understand the power of local consistency
methods for hybrid CSPs. On this quest, we focus in this article on
the power of the first level of local consistency, known as \emph{arc
consistency} (AC), for classes of binary hybrid CSPs defined by
forbidden (partially-ordered) patterns.

The class of CSPs defined by forbidding BTP from
Figure~\ref{fig:btpmc}(a) on page~\pageref{fig:btpmc} is in fact solvable by AC. But as it turns out,
BTP is not the only pattern with this property.

As our main contribution, we give, in Theorem~\ref{thm:order}, a
\emph{complete characterisation} of so-called simple
partially-ordered forbidden patterns which have this AC-solvability
property. Here the partial orders are on variables and domain values.
It turns out that BTP is just one of five such AC-solvable patterns.
The four other patterns allow us to exhibit new tractable classes,
one of which in particular we expect to lead to new applications
since it defines a strict generalisation of binary max-closed
constraints which have already found applications in computer
vision~\cite{Cooper99:lines} and temporal
reasoning~\cite{Dechter91:temporal}. We also provide results on the
associated meta problem of deciding whether a CSP instance falls into
one of these new tractable classes.

Given that AC is the first level of local consistency
methods\footnote{In some AI literature AC is the second level, the
first being \emph{node consistency}~\cite{Rossi06:handbook}. AC is
also the first level for \emph{relational
width}~\cite{Bulatov06:ja}.} and is implemented in \emph{all}
constraint solvers, an understanding of the power of AC is paramount.
We note that focusing on classes of CSPs defined by forbidden
patterns is very natural as AC \emph{cannot} introduce forbidden
patterns because pattern occurrence is defined
by the presence of some compatibility and incompatibility pairs
between values, and by definition, such pairs cannot be added when
values are removed. While simple patterns do not cover all partially-ordered
patterns it is a natural, interesting, and broad enough concept that
covers BTP and four other novel and non-trivial tractable classes. We
expect our results and techniques to be used in future work on the
power of AC.

\subsection*{Related work}

Computational complexity classifications have been obtained for binary CSPs
defined by forbidden negative patterns (i.e., only pairwise incompatible
assignments are specified)~\cite{cccms12:jair} and for binary CSPs defined by
patterns on 2 constraints~\cite{Cooper15:dam}. Moreover, (generalisations of)
forbidden patterns have been studied in the context of variable and 
 value elimination rules~\cite{ccez15:jcss}. Finally, the idea of forbidding patterns
as topological minors has recently been investigated~\cite{ccjz15:ijcai}.

\cite{Kolmogorov15:arxiv,Takhanov15:arxiv} recently considered the possible
extensions of the algebraic approach from the language to the hybrid setting.

The power of the valued version of AC~\cite{Cooper10:osac} has recently been
characterised~\cite{ktz15:sicomp}. Moreover, the valued version of AC
is known to solve all tractable finite-valued language-based
CSPs~\cite{tz16:jacm}.

\section{Preliminaries}
\label{sec:prelim}

\subsection{CSPs and patterns}
\label{sec:patt}

A pattern can be seen as a generalisation of the concept of a binary CSP
instance that leaves the consistency of some assignments to pairs of variables
undefined.

\begin{definition}
A \emph{pattern} is a four-tuple \tuple{X, D, A, \cpt} where:
\begin{itemize}
\item $X$ is a finite set of \emph{variables};
\item $D$ is a finite set of \emph{values};
\item $A \subseteq X \times D$ is the set of possible variable-value assignments
called \emph{points};
the \emph{domain} of $x \in X$ is its non-empty set $D(x)$
of possible  values: $D(x) = \{a \in D \mid \tuple{x,a} \in
A\}$;
\item $\cpt$ is a partial \emph{compatibility function} from the set
of unordered pairs of points \[\{\{\tuple{x,a},\tuple{y,b}\} \mid x
\neq y\}\] to $\{{\tt TRUE},{\tt FALSE}\}$. If
$\cpt(\tuple{x,a},\tuple{y,b})$ $=$ {\tt TRUE} (resp., {\tt FALSE})
we say that \tuple{x,a} and \tuple{y,b} are \emph{compatible} (resp.,
\emph{incompatible}). For simplicity, we write $\cpt(p,q)$ for
$\cpt(\{p,q\})$.
\end{itemize}
\end{definition}

We will use a simple figurative drawing for patterns. Each variable will be
drawn as an oval containing dots for each of its possible points. Pairs in
the domain of the function $\cpt$ will be represented by lines between points:
solid lines (called \emph{positive}) for compatibility and dashed lines (called \emph{negative}) for
incompatibility.

\begin{example}
The pattern in Figure~\ref{fig:lx}  on page~\pageref{fig:lx} is called LX. It consists of
three variables, five points, six positive edges, and two negative edges.
\end{example}

We refine patterns to give a definition of a CSP instance.

\begin{definition}
A \emph{binary CSP instance} $P$ is a pattern \tuple{X, D, A, \cpt}
where \cpt\ is a total function, {\it i.e.} the domain of \cpt\ is
precisely $\{\{\tuple{x,a},\tuple{y,b}\} \mid x \neq y$, $a \in D(x)$, $b \in D(y)\}$.
\begin{itemize}
\item The \emph{relation} $R_{x,y} \subseteq D(x) \times D(y)$
on $\tuple{x,y}$ is $\{\tuple{a,b}\mid \cpt(\tuple{x,a},\tuple{y,b})
= {\tt TRUE}\}$.

\item A \emph{partial solution} to $P$ on $Y \subseteq X$ is a
mapping $s: Y \to D$ where, for all $x \neq y \in Y$ we have
$\tuple{s(x), s(y)} \in R_{x,y}$.

\item A \emph{solution} to $P$ is a partial solution on $X$.

\end{itemize}
\end{definition}

For notational simplicity we have assumed that there is \emph{exactly one}
binary constraint between each pair of variables. In particular, this means that
the absence of a constraint between variables $x,y$ is modelled by a complete
relation $R_{x,y} = D(x) \times D(y)$ allowing every possible pair of
assignments to $x$ and $y$.
We say that there is a \emph{non-trivial} constraint on variables $x,y$ if
$R_{x,y} \neq D(x) \times D(y)$. We also use the simpler notation $R_{ij}$ for $R_{x_i,x_j}$.

The main focus of this paper is on ordered patterns, which additionally allow
for variable and value orders.

\begin{definition}
An \emph{ordered pattern} is a six-tuple \tuple{X, D, A, \cpt, <_X, <_D} where:
\begin{itemize}
\item \tuple{X,D,A,\cpt} is a pattern;
\item $<_X$ is a (possibly partial) strict order on $X$; and
\item $<_D$ is a (possibly partial) strict order on $D$.
\end{itemize}
\end{definition}

A pattern  $\tuple{X,D,A,\cpt}$ can be seen as an ordered pattern with empty variable and value orders, i.e.
$\tuple{X,D,A,\cpt,\emptyset,\emptyset}$.

Throughout the paper when we say ``pattern'' we implicitly mean
``ordered pattern'' and use the word ``unordered'' to emphasize, if
needed, that the pattern in question is not ordered.

We do not consider patterns with structure (such as equality or order) between elements in the domains
of {\em distinct} variables.

\begin{definition}
A pattern $P=\tuple{X,D,A,\cpt,<_X,<_D}$  is called \emph{basic} if
(1) $D(x)$ and $D(y)$ do not intersect for distinct
$x,y\in X$, and
(2) $<_D$ only contains pairs of elements
$\tuple{a,b}$ from the domain of the same variable, i.e., $a,b\in D(x)$ for some
$x\in X$.
\end{definition}

\begin{example}
The pattern shown in Figure~\ref{fig:btpmc}(a) 
is known as the \emph{broken triangle}
pattern (BTP)~\cite{cjs10:aij-btp}. BTP consists of three variables, four
points, three positive edges, two negative edges, $<_X=\{x<z,y<z\}$, and
$<_D=\emptyset$. Given a basic pattern, we can refer to a point
$\tuple{x,a}$ in the pattern as simply $a$ when the variable is clear from the context or a figure.
For instance, the point $\tuple{z,\gamma}$ in Figure~\ref{fig:btpmc}(a) 
can be referred to as $\gamma$.
\end{example}

\begin{example}
The pattern in Figure~\ref{fig:btpmc}(b) 
is the (binary) \emph{max-closed}
pattern (MC). The pattern MC consists of two variables, four points, two
positive edges, one negative edge, $<_X=\emptyset$, and
$<_D=\{\beta<\alpha,\delta<\gamma\}$.
MC (Figure~\ref{fig:btpmc}(b)) 
together with the extra structure $\alpha>\gamma$ is an example of a pattern that is not basic.
\end{example}

\begin{figure}   
\centering
\begin{picture}(400,100)(0,0)

\put(0,0){
\begin{picture}(170,100)(0,0)
\put(10,60){\usebox{\varone}} \put(50,10){\usebox{\varone}}
\put(90,50){\usebox{\vartwo}}
\dashline{5}(20,80)(100,80) \dashline{5}(60,30)(100,60)
\put(20,80){\line(4,-1){80}}
\put(20,80){\line(4,-5){40}} \put(60,30){\line(4,5){40}}
\put(60,10){\makebox(0,0){$x$}}
  \put(20,58){\makebox(0,0){$y$}}  \put(100,42){\makebox(0,0){$z$}}
  \put(140,60){\makebox{$x,y<z$}}
  \put(115,80){\makebox(0,0){$\gamma$}}
  \put(115,60){\makebox(0,0){$\delta$}}
  \put(15,10){\makebox(0,0){(a)}}
\end{picture}}

\put(220,0){
\begin{picture}(180,100)(0,0)
\put(10,50){\usebox{\vartwo}} \put(90,50){\usebox{\vartwo}}
\dashline{5}(20,80)(100,80)
\put(20,80){\line(4,-1){80}} \put(20,60){\line(4,1){80}}
  \put(20,42){\makebox(0,0){$x$}}  \put(100,42){\makebox(0,0){$y$}}
  \put(140,60){\makebox{\shortstack{$\alpha>\beta$ \\ $\gamma>\delta$}}}
  \put(5,80){\makebox(0,0){$\alpha$}}
  \put(5,60){\makebox(0,0){$\beta$}}
  \put(115,80){\makebox(0,0){$\gamma$}}
  \put(115,60){\makebox(0,0){$\delta$}}
  \put(15,10){\makebox(0,0){(b)}}
\end{picture}}

\end{picture}
\caption{Two AC-solvable patterns: (a) BTP (b) MC.}
\label{fig:btpmc}
\end{figure}
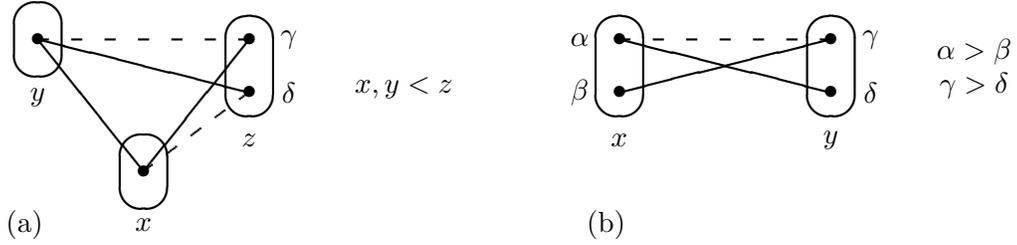  

For some of the proofs we will require patterns with additional structure,
namely, the ability to enforce certain points to be distinct.

\begin{definition}
A pattern with a disequality structure is a seven-tuple
\tuple{X, D, A, \cpt, <_X, <_D,\neq_D} where:
\begin{itemize}
\item \tuple{X,D,A,\cpt,<_X,<_D} is a pattern; and
\item $\neq_D\subseteq  D \times D$ 
is a set of pairs of domain values that are distinct. 
\end{itemize}
\end{definition}

\noindent An example of such a pattern is given in
Figure~\ref{fig:badpatterns}(b) on page~\pageref{fig:badpatterns}.

\subsection{Pattern occurrence}

Some points in a pattern are indistinguishable with respect to the rest of the
pattern.

\begin{definition} \label{def:mergeablepts}
Two points $a,b \in D(x)$ are \emph{mergeable} in a pattern
\mbox{\tuple{X,D,A,\cpt,<_X,<_D}} 
if there is no point $p \in A$ for which
$\cpt(\tuple{x,a},p)$, $\cpt(\tuple{x,b},p)$ are both defined and
$\cpt(\tuple{x,a},p) \neq \cpt(\tuple{x,b},p)$.
\end{definition}

\begin{definition}
A pattern is called \emph{unmergeable} if it does not contain any mergeable points.
\end{definition}

\begin{example}
The points $\gamma$ and $\delta$ in BTP (Figure~\ref{fig:btpmc}(a) on page~\pageref{fig:btpmc}) are not
mergeable since they have different compatibility with, for instance, the point
in variable $x$. The pattern LX (Figure~\ref{fig:lx} on page~\pageref{fig:lx}) is unmergeable.
\end{example}


Some points in a pattern (known as dangling points)
are redundant in arc-consistent CSP instances and hence can be removed.

\begin{definition}
Let $P=\tuple{X,D,A,\cpt,<_X,<_D}$ be a pattern. A point $p\in A$ is
called \emph{dangling} if it is not ordered by $<_D$ and if there is
at most one point $q\in A$ for which $\cpt(p,q)$ is defined, and
furthermore (if defined) $\cpt(p,q)= {\tt TRUE}$.
\end{definition}

\begin{example}
The point $\beta$ in the pattern MC (Figure~\ref{fig:btpmc}(b) on page~\pageref{fig:btpmc}) is not dangling since it is ordered.
\end{example}

In order to use (the absence of) patterns for AC-solvability we need
to define what we mean when we say that a pattern \emph{occurs} in a
CSP instance. We define the slightly more general notion of
occurrence of a pattern in another pattern, thus extending the
definitions for unordered patterns~\cite{Cooper15:dam}. Recall that a
CSP instance corresponds to the special case of a pattern whose
compatibility function is total. Essentially pattern $P$ occurs in
pattern $Q$ if $P$ is homomorphic to a subpattern of $Q$ via an
injective renaming of variables and a (possibly non-injective)
renaming of points~\cite{cccms12:jair}. We first make the observation
that dangling points in a pattern provide no useful information since
we assume that all CSP instances are arc consistent, which explains
why dangling points can be eliminated from patterns.

\begin{definition}  \label{def:simple}
A pattern is \emph{simple} if it is (i) basic, (ii) has no mergeable
points, and (iii) has no dangling points.
\end{definition}

From a given pattern it is possible to create an infinite number of
equivalent patterns by adding dangling points or by duplicating
points. By restricting our attention to simple patterns we avoid
having to consider such patterns.
We also discount non-basic patterns and mergeable patterns partly because of the
sheer number of cases to consider and partly because most of these patterns are
not very natural.

\begin{definition} \label{def:hom}
Let $P' = \tuple{X',D',A',\cpt',<_{X'},<_{D'}}$ and $P =
\tuple{X,D,A,\cpt,<_X,<_D}$ be two patterns. A \emph{homomorphism}
from $P'$ to $P$ is a mapping $f : A' \rightarrow A$ which satisfies:
\begin{itemize}
\item If $\cpt'(p,q)$ is defined, then $\cpt(f(p),f(q)) = \cpt'(p,q)$.
\item
The mapping $f_{var}: X' \rightarrow X$, given by $f_{var}(x') = x$
if $\exists a',a$ such that $f(\tuple{x',a'}) = \tuple{x,a}$, is
well-defined and injective.
\item If $x' <_{X'} y'$ then $f_{var}(x') <_{X} f_{var}(y')$.
\item If $a',b' \in D'(x')$, $a' <_{D'} b'$, $f(\tuple{x',a'})=\tuple{x,a}$
and  $f(\tuple{x',b'})=\tuple{x,b}$ then $a <_{D} b$.
\end{itemize}
\end{definition}

A \emph{consistent linear extension} of a pattern $P = \tuple{X,D,A,\cpt,<_X,<_D}$ is a  pattern
$P^t$ obtained from $P$
by first identifying any number of pairs of points $p,q$ which are both mergeable and incomparable 
(according to the transitive closure of $<_D$)
and then extending the orders on the variables and the domain values to total orders.

\begin{definition} \label{def:occ}
Let $P' = \tuple{X',D',A',\cpt',<_{X'},<_{D'}}$ and $P =
\tuple{X,D,A,\cpt,<_X,<_D}$ be two patterns.
$P'$ \emph{occurs} in $P$
if for all consistent linear extensions
$P^t$ of $P$, there is a homomorphism from $P'$ to $P^t$.
We use the notation CSP$_{\rm\overline{SP}}(P)$ to represent the set of
binary CSP instances in which the pattern $P$ does {\em not} occur.
\end{definition}


Thus, $I \in$ CSP$_{\rm\overline{SP}}(P)$ (the set of instances in which pattern $P$
does not occur as a subpattern) if one can find 
a total ordering of the variable-set and the domain of $I$ so that the pattern $P$ 
is not homomorphic to $I$ equipped with these orders.

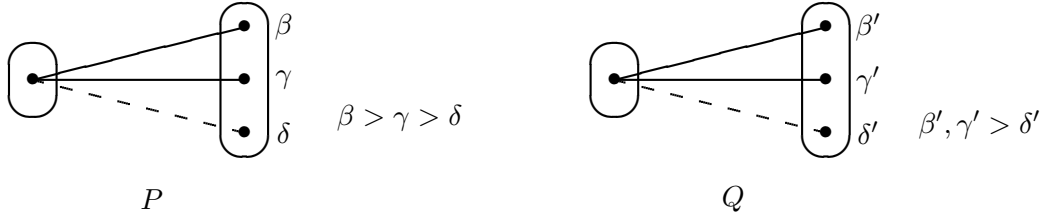
\begin{figure}  
\centering
\begin{picture}(400,80)(0,0)

\put(0,0){
\begin{picture}(170,85)(0,30)
\put(10,60){\usebox{\varone}} \put(90,50){\usebox{\varthree}}
\dashline{5}(20,80)(100,60) \put(20,80){\line(1,0){80}} \put(20,80){\line(4,1){80}}
  \put(135,63){\makebox{$\beta>\gamma>\delta$}}
  \put(115,100){\makebox(0,0){$\beta$}}
  \put(115,80){\makebox(0,0){$\gamma$}}
  \put(115,60){\makebox(0,0){$\delta$}}
  \put(65,35){\makebox(0,0){$P$}}
\end{picture}}

\put(220,0){
\begin{picture}(180,85)(0,30)
\put(10,60){\usebox{\varone}} \put(90,50){\usebox{\varthree}}
\dashline{5}(20,80)(100,60) \put(20,80){\line(1,0){80}}  \put(20,80){\line(4,1){80}}
  \put(135,60){\makebox{$\beta', \gamma' > \delta'$}}
  \put(116,100){\makebox(0,0){$\beta'$}}
  \put(116,80){\makebox(0,0){$\gamma'$}}
  \put(116,60){\makebox(0,0){$\delta'$}}
  \put(65,35){\makebox(0,0){$Q$}}
\end{picture}}

\end{picture}
\caption{Two patterns: $Q$ occurs in $P$, but $P$ does not occur in $Q$.}
\label{fig:occurrence}
\end{figure}  

The definition of occurrence extends in a natural way to patterns with a disequality structure.

\begin{rem}
We can add $a \neq b$ to a pattern, without changing its semantics,
when $b<_D a$ or
$\tuple{x,a}$ and $\tuple{x,b}$ are joined by negative and positive
edges to some point $\tuple{y,c}$. Furthermore, all domain values $a,b$ 
in an {\em instance} are distinct so there is an implicit $a \neq b$.
\end{rem}

\begin{example}
To illustrate the notion of consistent linear extension used in the definition of occurrence,
consider the two patterns $P,Q$ shown in Figure~\ref{fig:occurrence}.
For $P$ to occur in $Q$ we would require that there is a homomorphism from $P$ to $Q$ after any number of mergings
in the domain of $Q$ and any extension of the domain ordering of $Q$ to a total ordering: this effectively corresponds to 
the three cases $\beta'<\gamma'$, $\beta'=\gamma'$ and $\beta'>\gamma'$. Since there is not a homomorphism
from $P$ to the version of $Q$ in which $\beta'$ and $\gamma'$ have been merged, we can deduce that
$P$ does not occur in $Q$. On the other hand, $Q$ occurs in $P$ and, indeed, $Q$ occurs in
the pattern $V^>$ shown in Figure~\ref{fig:VV}(a) since $\beta'$ and $\gamma'$ can both map to $\gamma$.

As another example, the pattern MC (Figure~\ref{fig:btpmc}(b) on page~\pageref{fig:btpmc}) 
occurs in pattern EMC (Figure~\ref{fig:emc} on
page~\pageref{fig:emc})
but not in patterns BTP (Figure~\ref{fig:btpmc}(a) on
page~\pageref{fig:btpmc}) or BTX (Figure~\ref{fig:btx} on
page~\pageref{fig:btx}).
\end{example}

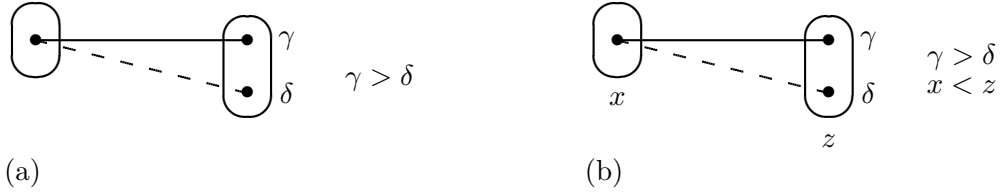
\begin{figure}  
\centering
\begin{picture}(400,85)(0,0)

\put(0,0){
\begin{picture}(170,80)(0,20)
\put(10,60){\usebox{\varone}} \put(90,50){\usebox{\vartwo}}
\dashline{5}(20,80)(100,60) \put(20,80){\line(1,0){80}}
  \put(137,63){\makebox{$\gamma>\delta$}}
  \put(115,80){\makebox(0,0){$\gamma$}}
  \put(115,60){\makebox(0,0){$\delta$}}
  \put(15,30){\makebox(0,0){(a)}}
\end{picture}}

\put(220,0){
\begin{picture}(180,80)(0,20)
\put(10,60){\usebox{\varone}} \put(90,50){\usebox{\vartwo}}
\dashline{5}(20,80)(100,60) \put(20,80){\line(1,0){80}}
  \put(20,58){\makebox(0,0){$x$}}  \put(100,42){\makebox(0,0){$z$}}
  \put(137,60){\makebox{\shortstack{$\gamma>\delta$ \\ $x<z$}}}
  \put(115,80){\makebox(0,0){$\gamma$}}
  \put(115,60){\makebox(0,0){$\delta$}}
  \put(15,30){\makebox(0,0){(b)}}
\end{picture}}

\end{picture}
\caption{Two subpatterns of EMC: (a) V$^>$ (b) V$^>_<$}
\label{fig:VV}
\end{figure}  

The main positive result in this paper is that CSP$_{\rm\overline{SP}}({\rm EMC})$ is solved by
arc consistency. In the following example we give generic examples of
instances in CSP$_{\rm\overline{SP}}({\rm EMC})$.

\begin{example}
CSP$_{\rm\overline{SP}}({\rm MC})$ is exactly the set of
binary CSP instances in which all constraints are max-closed~\cite{Jeavons95:maxclosed}.
Examples of such constraints are binary constraints of the form $y \geq F(x)$ for any function $F$
or constraints of the form $x \geq G(y)$ for any function $G$.

CSP$_{\rm\overline{SP}}({\rm EMC})$ includes all instances in CSP$_{\rm\overline{SP}}({\rm MC})$
but also other instances.
For example, if $I \in$ CSP$_{\rm\overline{SP}}({\rm MC})$, then we can create 
an instance $I' \in$ CSP$_{\rm\overline{SP}}({\rm EMC})$ which extends $I$ by adding 
some constraints which are not max-closed: in fact for each variable $y$ of $I$ we can add
new variables $z_y^1,\ldots,z_y^r$ in $I'$ with arbitrary constraints on $(y,z_y^i)$ ($i=1,\ldots,r$) provided that 
each $z_y^i$ ($i=1,\ldots,r$) is placed after $y$ in the variable ordering of $I'$ and each $z_y^i$
is constrained by no other variable. 

As another example of an instance $I''$  in CSP$_{\rm\overline{SP}}({\rm EMC})$,
again let $I \in$ CSP$_{\rm\overline{SP}}({\rm MC})$ and let $J \in$ CSP$_{\rm\overline{SP}}({\rm V}^>)$
where V$^>$ is the pattern shown in Figure~\ref{fig:VV}(a). Constraints in $J$ are of the form 
$x \leq F(y)$ where $F$ is an antitone function, i.e. $y_1 < y_2 \Rightarrow F(y_1) \geq F(y_2)$.
We combine $I$ and $J$ to create the instance $I''$. 
We first order the variables of $I,J$ so that all the variables of $I$ lie before
all the variables of $J$. To complete the description of $I''$, we add constraints between variables in $I$ and variables in $J$
so that the pattern $V^>_<$ (shown in Figure~\ref{fig:VV}(b)) does not occur with $x$ mapping
to a variable of $I$ and $z$ mapping to a variable of $J$. Such constraints are 
of the form $z \leq G(x)$ where $G$ is any function. The resulting instance $I''$
belongs to CSP$_{\rm\overline{SP}}({\rm EMC})$.
\end{example}

For a pattern $P$, we denote by unordered($P$) the underlying
unordered pattern, that is,
\[\mbox{unordered}(\tuple{X,D,A,\cpt,<_X,<_D})\ =\
\tuple{X,D,A,\cpt}.\]
For instance, the pattern unordered(BTP) is the
pattern from Figure~\ref{fig:btpmc}(a) on
page~\pageref{fig:btpmc} \emph{without} the structure
$x,y<z$.

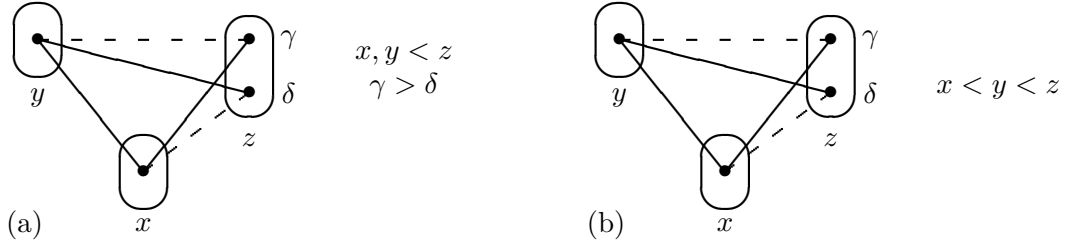
\begin{figure}   
\centering
\begin{picture}(400,110)(0,0)

\put(0,0){
\begin{picture}(170,100)(0,0)
\put(10,60){\usebox{\varone}} \put(50,10){\usebox{\varone}}
\put(90,50){\usebox{\vartwo}}
\dashline{5}(20,80)(100,80) \dashline{5}(60,30)(100,60)
\put(20,80){\line(4,-1){80}}
\put(20,80){\line(4,-5){40}} \put(60,30){\line(4,5){40}}
\put(60,10){\makebox(0,0){$x$}}
  \put(20,58){\makebox(0,0){$y$}}  \put(100,42){\makebox(0,0){$z$}}
  \put(140,60){\makebox{\shortstack{$x,y<z$ \\  $\gamma>\delta$}}}
  \put(115,80){\makebox(0,0){$\gamma$}}
  \put(115,60){\makebox(0,0){$\delta$}}
  \put(15,10){\makebox(0,0){(a)}}
\end{picture}}

\put(220,0){
\begin{picture}(170,100)(0,0)
\put(10,60){\usebox{\varone}} \put(50,10){\usebox{\varone}}
\put(90,50){\usebox{\vartwo}}
\dashline{5}(20,80)(100,80) \dashline{5}(60,30)(100,60)
\put(20,80){\line(4,-1){80}}
\put(20,80){\line(4,-5){40}} \put(60,30){\line(4,5){40}}
\put(60,10){\makebox(0,0){$x$}}
  \put(20,58){\makebox(0,0){$y$}}  \put(100,42){\makebox(0,0){$z$}}
  \put(140,60){\makebox{$x < y<z$}}
  \put(115,80){\makebox(0,0){$\gamma$}}
  \put(115,60){\makebox(0,0){$\delta$}}
  \put(15,10){\makebox(0,0){(b)}}
\end{picture}}

\end{picture}
\caption{Two equivalent versions of the broken triangle property:
forbidding the pattern (a) BTP$^{do}$ or forbidding the pattern (b) BTP$^{vo}$
defines the same class of instances.}
\label{fig:btp}
\end{figure}  

The occurrence relation between patterns is transitive.

\begin{lem} \label{lem:occ-transitive}
If $P$ occurs in $Q$ and $Q$ occurs in $R$, then $P$ occurs in $R$.
\end{lem}

\begin{proof}
Suppose that $P$ occurs in $Q$ and $Q$ occurs in $R$. Since $Q$ occurs in $R$, for all consistent
linear extensions $R^c$ of $R$, there exists a homomorphism $h: Q \rightarrow R^c$.
We will first construct a consistent linear extension $Q^c$ of $Q$ based on $R^c$ and $h$.
The set of variables of $R^c$ is totally ordered (by definition of a consistent linear extension).
Since, by Definition~\ref{def:hom}, the mapping $h_{var}$ induced by $h$ on the variables of $Q$ to the variables of $R^c$
is injective, there exists a total ordering $<_{X_{Q^c}}$ of the variables $X_Q$ of $Q$
corresponding to the total ordering $<_{X_{R^c}}$ in $R^c$ of the variables $h_{var}(X_Q)$
(i.e. $x <_{X_{Q^c}} y$ if and only if $h_{var}(x) <_{X_{R^c}} h_{var}(y)$). 

Now, consider any variable
$x$ in the pattern $Q$ and any pair of points $a,b$ in the domain $D_Q(x)$ of $x$ in $Q$.
First, consider the case when $a,b$ are non-mergeable.
By Definition~\ref{def:mergeablepts} and Definition~\ref{def:hom}, the points $h(\langle x,a \rangle), h(\langle x,b \rangle)$ are
also non-mergeable. For each such pair $a,b$, define $a <_{D_{Q^c}} b$ if and only if 
$a' <_{D_{R^c}} b'$ where $h(\langle x,a \rangle) =\langle x',a' \rangle$, $h(\langle x,b \rangle) =\langle x',b' \rangle$
and $<_{D_{R^c}}$ is the order on the domain values of $R^c$. This order is defined for each pair of
non-mergeable values $a,b$ since $R_c$ is a consistent linear extension.
Secondly, consider the case when $a,b$ are mergeable. If $h(\langle x,a \rangle) = h(\langle x,b \rangle)$
or if $h(\langle x,a \rangle), h(\langle x,b \rangle)$ are merged in the consistent linear extension $R^c$,
then we merge $a, b$ in $Q^c$; otherwise we define $a <_{D_{Q^c}} b$ if and only if 
$a' <_{D_{R^c}} b'$ where $h(\langle x,a \rangle) =\langle x',a' \rangle$, $h(\langle x,b \rangle) =\langle x',b' \rangle$.

Let $Q_c$ be identical to the pattern $Q$ but with the variable and domain orders $<_{X_{Q^c}}$ and $<_{D_{Q^c}}$
as defined above. By construction, $Q_c$ is a consistent linear extension of $Q$
and there exists a homomorphism $h': Q_c \rightarrow R_c$.
Since $P$ occurs in $Q$, we know that there exists a homomorphism $g: P \rightarrow Q^c$.
The composition of $g$ and $h'$ is then a homomorphism from $P$ to $R^c$. Since $R^c$
was an arbitrary linear extension of $R$, by Definition~\ref{def:occ}, $P$ occurs in $R$.
\end{proof}

In the following two simple lemmas, which follow from the definitions, $P$,
$Q$, and $R$ are patterns and $I$ is an instance.

\begin{lem} \label{lem:occ-sup}
If $P$ occurs in $Q$ and $P$ does not occur in $I$, then $Q$ does not occur in
$I$, i.e. CSP$_{\overline{SP}}(P)$ $\subseteq$  CSP$_{\overline{SP}}(Q)$.
\end{lem}

\begin{lem} \label{lem:occ-unordered}
For any pattern $P$, unordered($P$) occurs in $P$.
\end{lem}

\subsection{AC solvability}

Arc consistency (AC) is a fundamental concept for CSPs.

\begin{definition}
Let $I = \tuple{X,D,A,\cpt}$ be a CSP instance. A point
$\tuple{x,a} \in A$ is called \emph{arc consistent}
if, for all  variables $y \neq x$ in $X$ there is some point
$\tuple{y,b} \in A$ compatible  with \tuple{x,a}.

The CSP instance \tuple{X,D,A,\cpt} is called \emph{arc consistent}
if $A \neq \emptyset$ and every point in $A$ is arc consistent.
\end{definition}

Points that are not arc-consistent cannot be part of a solution so can
safely be removed. There are optimal $O(cd^2)$ algorithms for establishing arc
consistency which repeatedly remove such points~\cite{Bessiere:AC}, where $c$ is
the number of non-trivial constraints and $d$ the maximum domain size.
Algorithms establishing arc consistency are implemented in all constraint solvers.

AC is a \emph{decision procedure} for a class of CSP instances if for every
instance from the class, after establishing arc consistency, non-empty domains
for all variables guarantee the existence of a solution to the instance. (Note
that a solution can then be found without backtrack by maintaining AC during
search). 


\begin{definition}
A pattern $P$ is called \emph{AC-solvable} if AC is a decision procedure for
CSP$_{\overline{SP}}(P)$.
\end{definition}

The following lemma is a straightforward consequence of the definitions.

\begin{lem} \label{lem:not-ac}
A pattern  $P$ is not AC-solvable if and only if
there is an instance $I \in$ CSP$_{\overline{SP}}(P)$ that is arc consistent
and has no solution.
\end{lem}

The following lemma follows directly from Lemmas~\ref{lem:occ-sup} and \ref{lem:not-ac}.

\begin{lem} \label{lem:hered}
If $P$ occurs in $Q$ and $P$ is not AC-solvable, then $Q$ is not AC-solvable.
\end{lem}

As our main result we will, in Theorem~\ref{thm:order}, characterise
all simple patterns that are AC-solvable.

\subsection{Pattern symmetry and equivalence}

For an ordered pattern $P$, we denote by invDom($P$), invVar($P$)
the patterns obtained from $P$ by inversing
the domain order or the variable order, respectively.

\begin{lem} \label{lem:inv}
If $P$ is not AC-solvable, then neither is any of invDom($P$), invVar($P$) or
invDom(invVar($P$)).
\end{lem}
\begin{proof}
For an arc-consistent instance the property of having a solution is
independent of any variable or domain orderings.
The claims follow from inversing the respective orders in the
instance $I$ of Lemma~\ref{lem:not-ac} proving that $P$ is not AC-solvable.
\end{proof}

Some patterns define the same classes of CSP instances.

\begin{definition}
Patterns $P$ and $P'$ are \emph{equivalent} if
\[
\mbox{CSP}_{\overline{SP}}(P)\ =\ \mbox{CSP}_{\overline{SP}}(P').
\]
\end{definition}

\begin{lem}
If $P$ occurs in $P'$ and $P'$ occurs in $P$, then $P,P'$ are
equivalent.
\end{lem}

\begin{example}\label{ex:lx<}
Let LX$^<$ be the pattern obtained from LX (Figure~\ref{fig:lx} on
page~\pageref{fig:lx}) by adding the
partial variable order $y<z$. Due to the symmetry of LX, observe that LX and
LX$^<$ are equivalent.
\end{example}

\begin{example}
The two patterns shown in Figure~\ref{fig:btp} on
page~\pageref{fig:btp} are also equivalent: (a)
BTP$^{do}$ with structure $x,y < z$ and $c < d$, and (b) BTP$^{vo}$ with
variable order $x < y < z$. We will call these the \emph{variable-ordered} and
\emph{domain-ordered} versions of BTP, respectively, when it is necessary to make the
distinction between the two. BTP (Figure~\ref{fig:btpmc}(a) on
page~\pageref{fig:btpmc}) will refer to the same pattern with the only
structure $x,y<z$ which again, by symmetry, is equivalent to both
BTP$^{do}$ and  BTP$^{vo}$.
\end{example}

\section{New tractable classes solved by arc consistency}

Our search for a characterisation of all simple patterns decided by
arc consistency surprisingly uncovered four new tractable patterns,
which we describe in this section. The first pattern we study is
shown in Figure~\ref{fig:emc}. 
It is a proper generalisation of the
MC pattern (Figure~\ref{fig:btpmc}(b) on page~\pageref{fig:btpmc}) since it has an extra variable
and three extra edges.

\begin{figure}
\centering
\begin{picture}(180,100)(0,0)

\put(10,50){\usebox{\vartwo}} \put(50,10){\usebox{\varone}}
\put(90,50){\usebox{\vartwo}}
\dashline{5}(20,80)(100,80) \dashline{5}(60,30)(100,60)
\put(20,80){\line(4,-1){80}} \put(20,60){\line(4,1){80}}
\put(20,80){\line(4,-5){40}} \put(60,30){\line(4,5){40}}
\put(60,10){\makebox(0,0){$x$}}
  \put(20,42){\makebox(0,0){$y$}}  \put(100,42){\makebox(0,0){$z$}}
  \put(160,60){\makebox{\shortstack{$y<z$ \\ $\alpha>\beta$ \\ $\gamma>\delta$}}}
  \put(5,80){\makebox(0,0){$\alpha$}}
  \put(5,60){\makebox(0,0){$\beta$}}
  \put(115,80){\makebox(0,0){$\gamma$}}
  \put(115,60){\makebox(0,0){$\delta$}}
\put(75,30){\makebox(0,0){$\epsilon$}}

\end{picture}
\caption{The ordered pattern EMC (Extended Max-Closed)}
\label{fig:emc}
\end{figure}
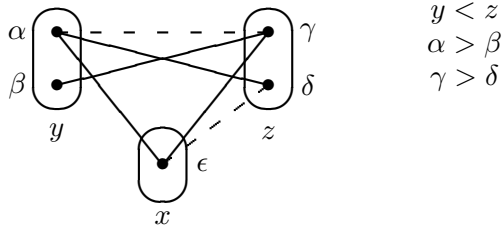

\begin{thm} \label{thm:emc}
AC is a decision procedure for CSP$_{\overline{SP}}(EMC)$ where EMC is
the pattern shown in Figure~\ref{fig:emc}. 
\end{thm}

\begin{proof}
Since establishing arc consistency only eliminates domain elements,
and hence cannot introduce the pattern, it suffices to show that
every arc-consistent instance $I=\tuple{X, D, A, \cpt}$ $\in$
CSP$_{\overline{SP}}(EMC)$ has a solution. We give a constructive
proof. For a given domain ordering, let $x_1 < \ldots < x_n$ be an ordering of $X$
such that EMC does not occur in $I$. 
Define an assignment $\tuple{a_1,\ldots,a_n}$ to the variables
$\tuple{x_1,\ldots,x_n}$ recursively as follows: $a_1 = \max(D(x_1))$
and, for $i>1$,
\begin{equation} \label{eq:defai}
a_i \ = \ \min \{ a_i^j \mid 1 \leq j < i \}, \ \
\ {\rm where} \ \ \ a_i^j  = \max \{a \in D(x_i) \mid (a_j,a) \in R_{ji}\}
\end{equation}
In other words, given an assignment $\tuple{a_1,\ldots,a_{i-1}}$ to variables
$\tuple{x_1,\ldots,x_{i-1}}$, for $1 \leq j < i$, $a_i^j$ is the maximum value
in $D(x_i)$ compatible with $a_j$. Then $a_i$ is the minimum of these values 
$a_i^1,\ldots,a_i^{i-1}$.
We prove by strong induction on $n$ the following claim. 
\begin{description}
\item[$C(n)$] $\tuple{a_1,\ldots,a_n}$
defined by Equation (\ref{eq:defai}) is a solution to arc-consistent instances $I \in$ 
CSP$_{\overline{SP}}(EMC)$ of size $n$.
\end{description}
The claim trivially holds for $n=1$ since $a_1 \in D(x_1)$. It remains to show that
if the claim holds for instances of size less than $n$ then it holds for instances of size $n$.

We suppose that $C(1) \wedge \ldots \wedge C(n-1)$ holds and we will
show that this implies $C(n)$, i.e. that $\tuple{a_1,\ldots,a_n}$ is a solution to 
instance $I \in$  CSP$_{\overline{SP}}(EMC)$ of size $n$.
Suppose, for a contradiction, that $(a_j,a_k) \notin R_{jk}$ for some
$1 \leq j < k \leq n$. If there is more than one such pair $(j,k)$,
then choose $k$ to be minimal and then for this value of $k$ choose
$j$ to be minimal.

For $i>1$, we denote by $pred(i)$ a 
value of $j<i$ such that $a_i = a_i^j$. Arc consistency guarantees
that $a_i^j$ exists and hence that $a_i$ and $pred(i)$ are well
defined. 
Let $m_0=k$ and $m_r = pred(m_{r-1})$ for $r \geq 1$ if $m_{r-1} > 1$.
Let $t$ be such that $m_t=1$. By definition of $pred$, we have
$$1 = m_t  \ < \ m_{t-1} \ < \ \ldots \ < \ m_1 \ < \ m_0 = k$$
which implies that this series is finite and hence that $t$ is well-defined.
We will show that for $r=1,\ldots,t$, $\exists b_r \in D(x_{m_r})$ such that $b_r > a_{m_r}$,
which for $r=t$ will provides us with the contradiction we are looking for since $a_{m_t} = a_1 = \max(D(x_1))$.

We distinguish two cases: (1) $j > m_1$, and (2) $j < m_1$. Since
$(a_j,a_k) \notin R_{jk}$ and $(a_{m_1},a_k) \in R_{jk}$ we know that $j \neq m_1$.

\paragraph{Case (1) $j > m_1$:} Define $b_0 = a_k^{j}$. By definition of $a_k$,
we know that $a_k \leq a_k^{j}$. Since $(a_j,a_k) \notin R_{jk}$
and $(a_j,a_k^{j}) \in R_{jk}$, we have $b_0 = a_k^{j} > a_k$.

By our choice of $j$ to be minimal, and since $j > m_1$ we know that $(a_{m_r},a_k) \in R_{m_r k}$
for $r=1,\ldots,t$. Indeed, by minimality of $k$, we already had
$(a_{m_r},a_{m_s}) \in R_{m_r m_s}$ for $1 \leq s \leq r \leq t$. Thus, since $k=m_0$,
we have
\begin{equation}
(a_{m_r},a_{m_s}) \in R_{m_r m_s} \ \ {\rm for} \ \ 0 \leq s \leq r \leq t.  \label{eq:posedges}
\end{equation}

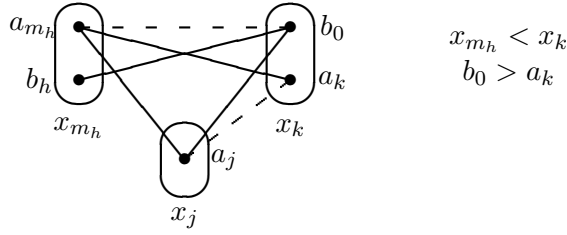
\begin{figure}
\centering
\begin{picture}(180,100)(0,0)

\put(10,50){\usebox{\vartwo}} \put(50,10){\usebox{\varone}}
\put(90,50){\usebox{\vartwo}}
\dashline{5}(20,80)(100,80) \dashline{5}(60,30)(100,60)
\put(20,80){\line(4,-1){80}} \put(20,60){\line(4,1){80}}
\put(20,80){\line(4,-5){40}} \put(60,30){\line(4,5){40}}
\put(60,8){\makebox(0,0){$x_j$}}
  \put(20,42){\makebox(0,0){$x_{m_h}$}}  \put(100,42){\makebox(0,0){$x_k$}}
  \put(160,60){\makebox{\shortstack{$x_{m_h}<x_k$ \\ $b_0 > a_k$}}}
  \put(3,80){\makebox(0,0){$a_{m_h}$}}
 \put(5,60){\makebox(0,0){$b_h$}}
  \put(116,80){\makebox(0,0){$b_0$}}
  \put(116,60){\makebox(0,0){$a_k$}}
\put(75,30){\makebox(0,0){$a_j$}}

\end{picture}
\caption{To avoid the pattern EMC, we must have $b_h > a_{m_h}$.}
\label{fig:almostemc}
\end{figure}

By arc consistency, $\exists b_1 \in D(x_{m_1})$ such that $(b_1,b_0) \in R_{m_1 k}$.
We have $(a_{m_1},a_j) \in R_{m_1 j}$ by minimality of $k$ and since $m_1,j < k$.
Since $m_1 = pred(k)$ and hence $a_k = a_k^{m_1}$, we have $(a_{m_1},a_k) \in R_{m_1 k}$
and $(a_{m_1}, b_0) \notin R_{m_1 k}$ by the maximality of $a_k^{m_1}$ in Equation~(\ref{eq:defai}) 
and since $b_0 > a_k = a_k^{m_1}$.
We thus have the situation illustrated in Figure~\ref{fig:almostemc} 
for $h=1$. Since the pattern EMC
does not occur in $I$, we must have $b_1 > a_{m_1}$.

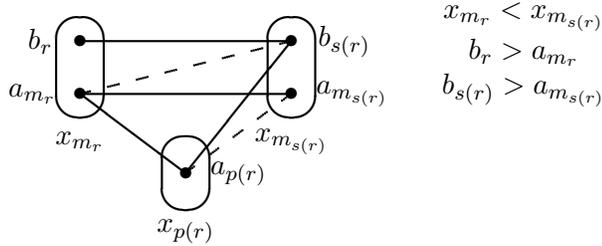
\begin{figure}
\centering
\begin{picture}(180,100)(0,0)

\put(10,50){\usebox{\vartwo}} \put(50,10){\usebox{\varone}}
\put(90,50){\usebox{\vartwo}}
\dashline{5}(20,60)(100,80) \dashline{5}(60,30)(100,60)
\put(20,60){\line(1,0){80}} \put(20,80){\line(1,0){80}}
\put(20,60){\line(4,-3){40}} \put(60,30){\line(4,5){40}}
\put(60,8){\makebox(0,0){$x_{p(r)}$}}
  \put(20,42){\makebox(0,0){$x_{m_r}$}}  \put(100,42){\makebox(0,0){$x_{m_{s(r)}}$}}
  \put(157,60){\makebox{\shortstack{$x_{m_r}<x_{m_{s(r)}}$ \\ $b_r > a_{m_r}$ \\ $b_{s(r)} > a_{m_{s(r)}}$}}}
  \put(5,80){\makebox(0,0){$b_r$}}
  \put(2,60){\makebox(0,0){$a_{m_r}$}}
  \put(120,80){\makebox(0,0){$b_{s(r)}$}}
  \put(123,60){\makebox(0,0){$a_{m_{s(r)}}$}}
\put(80,30){\makebox(0,0){$a_{p(r)}$}}

\end{picture}
\caption{The situation corresponding to hypothesis $H_r$.}
\label{fig:Hremc}
\end{figure}

For $1 \leq r \leq t$, let $H_r$ be the following hypothesis.
\begin{description}
\item[$H_r$]  $\exists s(r) \in\{0,\ldots,r-1\}$,
$\exists p(r) < k$, $\exists b_r \in D(x_{m_r})$, with $b_r >
a_{m_r}$, such that we have the situation shown in
Figure~\ref{fig:Hremc}. 
\end{description}
We have just shown that $H_1$ holds (with $s(1)=0$ and $p(1)=j$).
We now show, for $1 \leq r < t$, that $(H_1 \wedge \ldots \wedge H_r) \Rightarrow H_{r+1}$.

We know that $(a_{m_{r+1}},a_{m_r}) \in R_{m_{r+1} m_r}$ and
$(a_{m_{r+1}},b_{r}) \notin R_{m_{r+1} m_r}$, since $m_{r+1} =
pred(m_r)$  and by maximality of $a_{m_r} = a_{m_r}^{m_{r+1}}$ in
Equation~(\ref{eq:defai}). Let $q \in \{0,\ldots,r\}$ be minimal such
that $(a_{m_{r+1}},b_q) \notin R_{m_{r+1} m_q}$. We distinguish two
cases: (a) $q=0$, and (b) $q>0$.

If $q=0$, then we have $(a_{m_{r+1}},a_k) \in R_{m_{r+1} k}$ (from
Equation~(\ref{eq:posedges}), since $k=m_0$), $(a_{m_{r+1}},b_0)
\notin R_{m_{r+1} k}$ (since $q=0$), $(a_{m_{r+1}},a_j) \in
R_{m_{r+1} j}$ (by minimality of $k$, since $m_{r+1},j < k$). By arc
consistency, $\exists b_{r+1} \in D(x_{m_{r+1}})$ such that
$(b_{r+1},b_0) \in R_{m_{r+1} k}$. We then have the situation
illustrated in Figure~\ref{fig:almostemc} 
for $h=r+1$. As above, from
the absence of pattern EMC, we can deduce that $b_{r+1} >
a_{m_{r+1}}$. We thus have $H_{r+1}$ (with $s(r+1)=0$ and
$p(r+1)=j$).

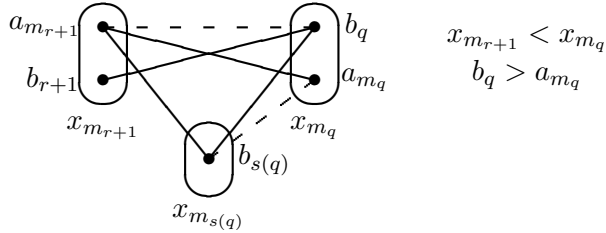
\begin{figure}
\centering
\begin{picture}(190,100)(-10,0)

\put(10,50){\usebox{\vartwo}} \put(50,10){\usebox{\varone}}
\put(90,50){\usebox{\vartwo}}
\dashline{5}(20,80)(100,80) \dashline{5}(60,30)(100,60)
\put(20,80){\line(4,-1){80}} \put(20,60){\line(4,1){80}}
\put(20,80){\line(4,-5){40}} \put(60,30){\line(4,5){40}}
\put(60,8){\makebox(0,0){$x_{m_{s(q)}}$}}
  \put(20,42){\makebox(0,0){$x_{m_{r+1}}$}}  \put(100,42){\makebox(0,0){$x_{m_q}$}}
  \put(150,60){\makebox{\shortstack{$x_{m_{r+1}}<x_{m_q}$ \\ $b_q > a_{m_q}$}}}
  \put(-2,80){\makebox(0,0){$a_{m_{r+1}}$}}
 \put(1,60){\makebox(0,0){$b_{r+1}$}}
  \put(116,80){\makebox(0,0){$b_q$}}
  \put(119,60){\makebox(0,0){$a_{m_q}$}}
\put(81,30){\makebox(0,0){$b_{s(q)}$}}

\end{picture}
\caption{To avoid the pattern EMC, we must have $b_{r+1} > a_{m_{r+1}}$.}
\label{fig:almostemc2}
\end{figure}

If $q>0$, then $H_1 \wedge \ldots \wedge H_r$ implies that $H_q$
holds. By minimality of $q$, we know that $(a_{m_{r+1}}, b_{s(q)})
\in R_{m_{r+1} m_{s(q)}}$ since $s(q) < q$. We know that
$(a_{m_{r+1}}, a_{m_q}) \in R_{m_{r+1} m_q}$ from
Equation~(\ref{eq:posedges}), and that $(a_{m_{r+1}},b_q) \notin
R_{m_{r+1} m_q}$ by definition of $q$. We know that $(b_q,b_{s(q)})
\in R_{m_q m_{s(q)}}$ and $(a_{m_q},b_{s(q)}) \notin R_{m_q
m_{s(q)}}$ from $H_q$. By arc consistency,  $\exists b_{r+1} \in
D(x_{m_{r+1}})$ such that $(b_{r+1},b_q) \in R_{m_{r+1} m_q}$. We
then have the situation illustrated in Figure~\ref{fig:almostemc2}. 
As above, from the absence of pattern EMC, we can deduce that
$b_{r+1} > a_{m_{r+1}}$. We thus have $H_{r+1}$ (with $s(r+1)=q$ and
$p(r+1)=s(q)$).

\paragraph{Case (2) $j < m_1$:} Consider the subproblem $I'$ of $I$ on 
variables $\{x_1,x_2,\ldots,x_{m_1-1}\} \cup \{x_k\}$.
Since $x_{m_1}$ does not belong to the set of variables of $I'$, this
instance has size strictly less than $n$, and hence has a solution by our inductive
hypothesis $C(1) \wedge \ldots \wedge C(n-1)$ (i.e. 
that a solution is given by Equation~(\ref{eq:defai} for arc-consistent instances
in CSP$_{\overline{SP}}(EMC)$ of size less than $n$). The values of $a_i$ may differ between $I$
and $I'$. However, we can see from its definition given in
Equation~(\ref{eq:defai}), that the value of $a_i$ depends uniquely
on the subproblem on previous variables $\{x_1,\ldots,x_{i-1}\}$.
Showing the dependence on the instance by a superscript, we thus have
$a_i^{I'} = a_i^I$ ($i=1,\ldots,m_1-1$) although $a_k^{I'}$ may (and,
in fact, does) differ from $a_k^{I}$. By our inductive hypothesis $C(1) \wedge \ldots \wedge C(n-1)$,
$\tuple{a_1,\ldots,a_{m_1-1},a_k^{I'}}$ is a solution to $I'$.
Setting $b_0 = a_k^{I'}$, it follows that $(a_i,b_0) \in R_{ik}$ for
$1 \leq i < m_1$. In particular, since $j < m_1$, we have $(a_j,b_0)
\in R_{jk}$. Now $a_k^{I}  \leq a_k^{I'} = b_0$, since $I'$ is a
subinstance of $I$ (and so, from Equation~(\ref{eq:defai}), $a_k^{I}$
is the minimum of a superset over which $a_k^{I'}$ is a minimum).
Thus $a_k = a_k^{I} < b_0$, since $(a_j,b_0) \in R_{jk}$ and
$(a_j,a_k) \notin R_{jk}$.

By arc consistency, $\exists b_1 \in D(x_{m_1})$ such that $(b_1,b_0)
\in R_{m_1 k}$. As in case (1), we have the situation illustrated in
Figure~\ref{fig:almostemc} on
page~\pageref{fig:almostemc} for $h=1$. Since the pattern EMC does not
occur in $I$, we must have $b_1 > a_{m_1}$.

Consider the hypothesis $H_r$ stated in case (1) and illustrated in
Figure~\ref{fig:Hremc} on
page~\pageref{fig:Hremc}. We have just shown that $H_1$ holds (with
$s(1)=0$ and $p(1)=j$). We now show, for $1 \leq r < t$, that $(H_1
\wedge \ldots \wedge H_r) \Rightarrow H_{r+1}$.

As in case (1), we know that $(a_{m_{r+1}},a_{m_r}) \in R_{m_{r+1}
m_r}$ and $(a_{m_{r+1}},b_{r}) \notin R_{m_{r+1} m_r}$. Let $q \in
\{0,\ldots,r\}$ be minimal such that $(a_{m_{r+1}},b_q) \notin
R_{m_{r+1} m_q}$. We have seen above that $(a_{m_{r+1}},b_0) \in
R_{m_{r+1} k}$ (since $x_{m_{r+1}}$, $x_{m_0}$ are assigned,
respectively, the values $a_{m_{r+1}}$, $b_0$ in a solution to $I'$).
Therefore, we can deduce that $q>0$. Therefore $H_1 \wedge \ldots
\wedge H_r$ implies that $H_q$ holds. By minimality of $k$, and since
$m_q < m_0 = k$, we know that $(a_{m_{r+1}},a_{m_q}) \in R_{m_{r+1}
m_q}$. As in case (1), by minimality of $q$, we know that
$(a_{m_{r+1}}, b_{s(q)}) \in R_{m_{r+1} m_{s(q)}}$ . By arc
consistency,  $\exists b_{r+1} \in D(x_{m_{r+1}})$ such that
$(b_{r+1},b_q) \in R_{m_{r+1} m_q}$. We thus have the situation
illustrated in Figure~\ref{fig:almostemc2} on
page~\pageref{fig:almostemc2}. Again, from the absence
of pattern EMC, we can deduce that $b_{r+1} > a_{m_{r+1}}$. We thus
again have $H_{r+1}$ with $s(r+1)=q$ and $p(r+1)=s(q)$.

Thus, by induction on $r$, we have shown in both cases that $H_t$
holds. But recall that $m_t=1$ and that $a_1$ was chosen to be the
maximal element of $D(x_1)$ and hence $\nexists b_t \in D(x_1)$ such
that $b_t > a_1$. This contradiction shows that
$\tuple{a_1,\ldots,a_n}$ is a solution, as claimed.
\end{proof}

The next two patterns we study in this section, shown in Figure~\ref{fig:btx} 
and Figure~\ref{fig:bti}, 
are similar to EMC but the three patterns are incomparable (in the sense that none occurs in another)
due to the different orders on the three variables.

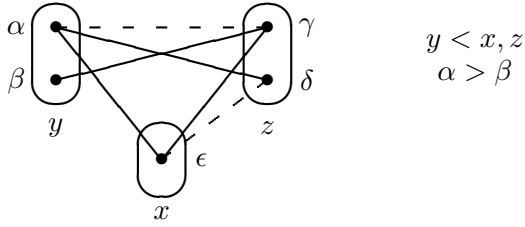
\begin{figure}
\centering
\begin{picture}(180,100)(0,0)

\put(10,50){\usebox{\vartwo}} \put(50,10){\usebox{\varone}}
\put(90,50){\usebox{\vartwo}}
\dashline{5}(20,80)(100,80) \dashline{5}(60,30)(100,60)
\put(20,80){\line(4,-1){80}} \put(20,60){\line(4,1){80}}
\put(20,80){\line(4,-5){40}} \put(60,30){\line(4,5){40}}
\put(60,10){\makebox(0,0){$x$}}
  \put(20,42){\makebox(0,0){$y$}}  \put(100,42){\makebox(0,0){$z$}}
  \put(160,60){\makebox{\shortstack{$y<x,z$ \\ $\alpha>\beta$}}}
  \put(5,80){\makebox(0,0){$\alpha$}}
  \put(5,60){\makebox(0,0){$\beta$}}
  \put(115,80){\makebox(0,0){$\gamma$}}
  \put(115,60){\makebox(0,0){$\delta$}}
\put(75,30){\makebox(0,0){$\epsilon$}}

\end{picture}
\caption{The ordered pattern BTX.}
\label{fig:btx}
\end{figure}

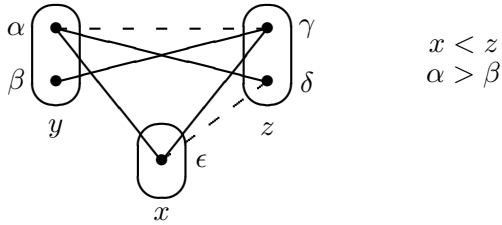
\begin{figure}
\centering
\begin{picture}(180,100)(0,0)

\put(10,50){\usebox{\vartwo}} \put(50,10){\usebox{\varone}}
\put(90,50){\usebox{\vartwo}}
\dashline{5}(20,80)(100,80) \dashline{5}(60,30)(100,60)
\put(20,80){\line(4,-1){80}} \put(20,60){\line(4,1){80}}
\put(20,80){\line(4,-5){40}} \put(60,30){\line(4,5){40}}
\put(60,10){\makebox(0,0){$x$}}
  \put(20,42){\makebox(0,0){$y$}}  \put(100,42){\makebox(0,0){$z$}}
  \put(160,60){\makebox{\shortstack{$x<z$ \\ $\alpha>\beta$}}}
  \put(5,80){\makebox(0,0){$\alpha$}}
  \put(5,60){\makebox(0,0){$\beta$}}
  \put(115,80){\makebox(0,0){$\gamma$}}
  \put(115,60){\makebox(0,0){$\delta$}}
\put(75,30){\makebox(0,0){$\epsilon$}}

\end{picture}
\caption{The ordered pattern BTI.}
\label{fig:bti}
\end{figure}

\begin{thm} \label{thm:btx}
AC is a decision procedure for CSP$_{\overline{SP}}(BTX)$ where BTX is
the pattern shown in Figure~\ref{fig:btx}. 
\end{thm}

\begin{proof}
Since establishing arc consistency only eliminates domain elements,
and hence cannot introduce the pattern, we only need to show that
every arc-consistent instance $I=\tuple{X,D,A,\cpt} \in$
CSP$_{\overline{SP}}(BTX)$ has a solution. For a given domain ordering, let $x_1 < \ldots< x_n$ be
an ordering of $X$ 
such that BTX does not occur in $I$. In fact we will show a stronger
result by proving that the hypothesis $H_n$, below, holds for all $n
\geq 1$.
\begin{description}
\item[$H_n$]   for all arc-consistent instances
$I=\tuple{X,D,A,\cpt}$ from CSP$_{\overline{SP}}(BTX)$ with $|X| =
n$, if $a = \max(D(x_1))$, then $I$ has a solution $s$ with $s(x_1)
= a$.
\end{description}
Trivially, $H_1$ holds. Suppose that $H_{n-1}$ holds where $n>1$. We
will show that this implies $H_n$, which will complete the proof by
induction. Let $D'(x_i) =\{ b \in D(x_i) \mid (a,b) \in R_{1i}\}$
($i=2,\ldots,n$). Denote by $I'$ the subproblem of $I$ on variables
$2,\ldots,n$ and domains $D'(x_i)$ ($i=2,\ldots,n$). To complete
the inductive proof, it is sufficient to show that $I'$ is arc
consistent: since $I' \in$ CSP$_{\overline{SP}}(BTX)$ and has $n-1$
variables, by $H_{n-1}$, if $I'$ is arc consistent it has a solution
which can clearly be extended to a solution to $I$ by adding the
assignment $(x_1,a)$.

To show arc consistency of $I'$, consider any variable $x_i$ with $i>1$ and any
$b\in D'(X_i)$. (We know by arc consistency of $I$ that $D'(x_i)$ is non-empty;
i.e. $(a,b)\in R_{1i}$.) Now let $x_j$ be any other variable with $j>1$. To
complete the proof, it suffices to show that $b$ has a support in $D'(x_j)$. 

By arc consistency of $I$, we can deduce the existence of $c \in D(x_j)$
such that $(b,c) \in R_{ij}$, and then $a' \in D(x_1)$ such that
$(a',c) \in R_{1j}$, as well as $d \in D(x_j)$ such that $(a,d) \in
R_{1j}$  (i.e. $d \in D'(x_j)$). If $b$ has no support in $D'(x_j)$,
then we must have $c \notin D'(x_j)$ (i.e. $(a,c) \notin R_{1j}$) and
$(b,d) \notin R_{ij}$. Since $a$ is the maximum element of $D(x_1)$,
we know that $a' \leq a$. Indeed, since $(a,c) \notin R_{1j}$ and
$(a',c) \in R_{1j}$, we have that $a' < a$. But then the pattern BTX
occurs in $I$, as shown in Figure~\ref{fig:occurrencebtx}. 
%
%
This contradiction
shows that $b$ does have a support in $D'(x_j)$ and hence that $I'$
is arc consistent, as required.
\end{proof}

\begin{figure}
\centering
\begin{picture}(180,100)(0,0)

\put(10,50){\usebox{\vartwo}} \put(50,10){\usebox{\varone}}
\put(90,50){\usebox{\vartwo}}
\dashline{5}(20,80)(100,80) \dashline{5}(60,30)(100,60)
\put(20,80){\line(4,-1){80}} \put(20,60){\line(4,1){80}}
\put(20,80){\line(4,-5){40}} \put(60,30){\line(4,5){40}}
\put(60,10){\makebox(0,0){$x_i$}}
  \put(20,42){\makebox(0,0){$x_1$}}  \put(100,42){\makebox(0,0){$x_j$}}
  \put(160,60){\makebox{\shortstack{$x_1<x_i,x_j$ \\ $a>a'$}}}
  \put(5,80){\makebox(0,0){$a$}}
  \put(5,60){\makebox(0,0){$a'$}}
  \put(115,80){\makebox(0,0){$c$}}
  \put(115,60){\makebox(0,0){$d$}}
\put(75,30){\makebox(0,0){$b$}}

\end{picture}
\caption{An occurrence of the pattern BTX.}
\label{fig:occurrencebtx}
\end{figure}
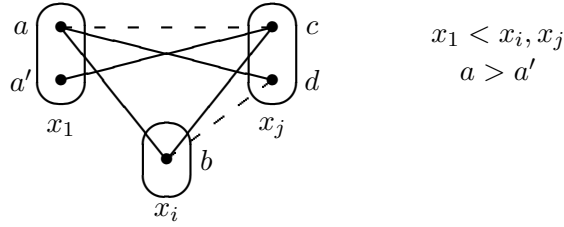

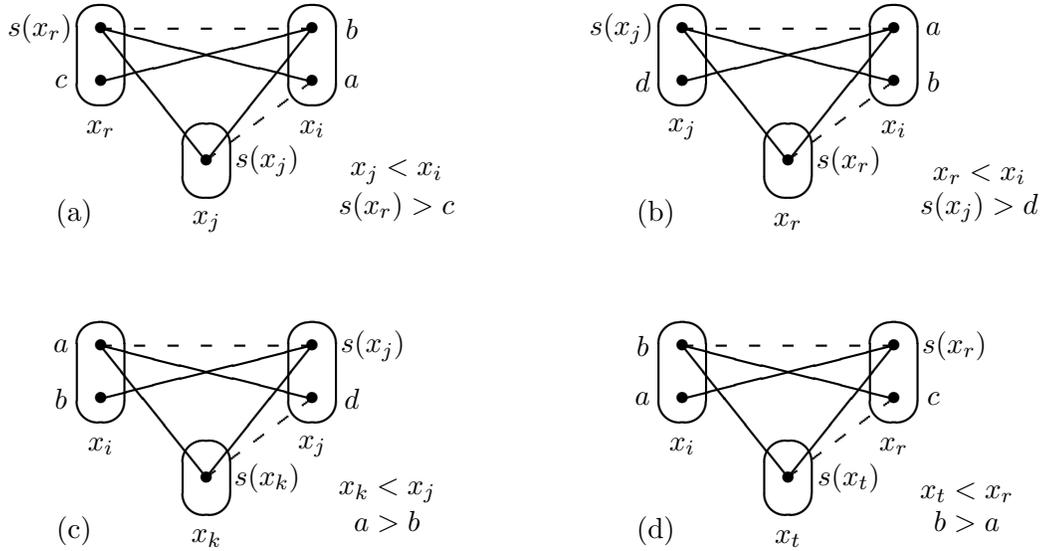
\begin{figure}
\centering
\begin{picture}(360,220)(0,0)

\put(0,120){
\begin{picture}(180,100)(0,0)
\put(10,50){\usebox{\vartwo}} \put(50,10){\usebox{\varone}}
\put(90,50){\usebox{\vartwo}}
\dashline{5}(20,80)(100,80) \dashline{5}(60,30)(100,60)
\put(20,80){\line(4,-1){80}} \put(20,60){\line(4,1){80}}
\put(20,80){\line(4,-5){40}} \put(60,30){\line(4,5){40}}
\put(60,7){\makebox(0,0){$x_j$}}
  \put(20,42){\makebox(0,0){$x_r$}}  \put(100,42){\makebox(0,0){$x_i$}}
  \put(110,10){\makebox{\shortstack{$x_j<x_i$ \\ $s(x_r)>c$}}}
  \put(-3,80){\makebox(0,0){$s(x_r)$}}
  \put(5,60){\makebox(0,0){$c$}}
  \put(115,80){\makebox(0,0){$b$}}
  \put(115,60){\makebox(0,0){$a$}}
\put(83,30){\makebox(0,0){$s(x_j)$}}
\put(10,10){\makebox(0,0){(a)}}
\end{picture}
}

\put(220,120){
\begin{picture}(180,100)(0,0)
\put(10,50){\usebox{\vartwo}} \put(50,10){\usebox{\varone}}
\put(90,50){\usebox{\vartwo}}
\dashline{5}(20,80)(100,80) \dashline{5}(60,30)(100,60)
\put(20,80){\line(4,-1){80}} \put(20,60){\line(4,1){80}}
\put(20,80){\line(4,-5){40}} \put(60,30){\line(4,5){40}}
\put(60,7){\makebox(0,0){$x_r$}}
  \put(20,42){\makebox(0,0){$x_j$}}  \put(100,42){\makebox(0,0){$x_i$}}
  \put(110,10){\makebox{\shortstack{$x_r<x_i$ \\ $s(x_j)>d$}}}
  \put(-3,80){\makebox(0,0){$s(x_j)$}}
  \put(5,60){\makebox(0,0){$d$}}
  \put(115,80){\makebox(0,0){$a$}}
  \put(115,60){\makebox(0,0){$b$}}
\put(83,30){\makebox(0,0){$s(x_r)$}}
\put(10,10){\makebox(0,0){(b)}}
\end{picture}
}

\put(0,0){
\begin{picture}(180,100)(0,0)
\put(10,50){\usebox{\vartwo}} \put(50,10){\usebox{\varone}}
\put(90,50){\usebox{\vartwo}}
\dashline{5}(20,80)(100,80) \dashline{5}(60,30)(100,60)
\put(20,80){\line(4,-1){80}} \put(20,60){\line(4,1){80}}
\put(20,80){\line(4,-5){40}} \put(60,30){\line(4,5){40}}
\put(60,7){\makebox(0,0){$x_k$}}
  \put(20,42){\makebox(0,0){$x_i$}}  \put(100,42){\makebox(0,0){$x_j$}}
  \put(110,10){\makebox{\shortstack{$x_k<x_j$ \\ $a>b$}}}
  \put(5,80){\makebox(0,0){$a$}}
  \put(5,60){\makebox(0,0){$b$}}
  \put(123,80){\makebox(0,0){$s(x_j)$}}
  \put(115,60){\makebox(0,0){$d$}}
\put(83,30){\makebox(0,0){$s(x_k)$}}
\put(10,10){\makebox(0,0){(c)}}
\end{picture}
}

\put(220,0){
\begin{picture}(180,100)(0,0)
\put(10,50){\usebox{\vartwo}} \put(50,10){\usebox{\varone}}
\put(90,50){\usebox{\vartwo}}
\dashline{5}(20,80)(100,80) \dashline{5}(60,30)(100,60)
\put(20,80){\line(4,-1){80}} \put(20,60){\line(4,1){80}}
\put(20,80){\line(4,-5){40}} \put(60,30){\line(4,5){40}}
\put(60,7){\makebox(0,0){$x_t$}}
  \put(20,42){\makebox(0,0){$x_i$}}  \put(100,42){\makebox(0,0){$x_r$}}
  \put(110,10){\makebox{\shortstack{$x_t<x_r$ \\ $b>a$}}}
  \put(5,80){\makebox(0,0){$b$}}
  \put(5,60){\makebox(0,0){$a$}}
  \put(123,80){\makebox(0,0){$s(x_r)$}}
  \put(115,60){\makebox(0,0){$c$}}
\put(83,30){\makebox(0,0){$s(x_t)$}}
\put(10,10){\makebox(0,0){(d)}}
\end{picture}
}

\end{picture}
\caption{Occurrences of the pattern BTI.}
\label{fig:occurrencesbti}
\end{figure}

\begin{thm} \label{thm:bti}
AC is a decision procedure for CSP$_{\overline{SP}}(BTI)$ where BTI is
the pattern shown in Figure~\ref{fig:bti}. 
\end{thm}

\begin{proof}
Since establishing arc consistency only eliminates domain elements,
and hence cannot introduce the pattern, we only need to show that
every arc-consistent instance $I=\tuple{X,D,A,\cpt} \in$
CSP$_{\overline{SP}}(BTI)$ has a solution. Let $x_1 < \ldots < x_n$
be an ordering of $X$ such that BTI does not occur in $I$. In fact,
we will show a stronger result by proving that, for all
arc-consistent instances $I=\tuple{X,D,A,\cpt} \in$
CSP$_{\overline{SP}}(BTI)$, $I$ has a solution $s$ defined
recursively by: $s(x_i)$ is the maximum value in $D(x_i)$ compatible
with all previous assignments $s(x_1),\ldots,s(x_{i-1})$.

Let $T_i$ denote the set of values in $D(x_i)$ compatible with the
assignments $s(x_1),\ldots,s(x_{i-1})$ (defined recursively as
above). If each $T_i$ ($i=1,\ldots,n$) is non-empty, then
$s(x_1),\ldots,s(x_{n})$ is a solution to $I$. Let $I_i$ denote the
subinstance of $I$ on the first $i$ variables $x_1,\ldots,x_i$. $T_1
= D(x_1)$ and is non-empty since $I$ is arc consistent, and hence
$s(x_1)=\max(T_1)$ is a solution to $I_1$. Suppose that
$s(x_1),\ldots,s(x_{i-1})$, as defined above, is a solution to
$I_{i-1}$. We will show that $s(x_i)$ exists and hence that
$s(x_1),\ldots,s(x_{i})$ is a solution to $I_{i}$, which by a simple
induction will complete the proof. Suppose for a contradiction that
$s(x_i)$ does not exist, i.e. that $T_i = \emptyset$.

For all $u$, let $R_{ji}(u)$ denote the subset of $D(x_i)$ which is compatible
with the assignment of $u$ to $x_j$, i.e. $R_{ji}(u) \ = \ \{v \in
D(x_i) \mid (u,v) \in R_{ji}\}$. For $j \in \{1,\ldots,i-1\}$, let
$T_i^j$ be the intersection of the sets $R_{hi}(s(x_h))$
($h=1,\ldots,j$). By our hypothesis that $T_i = \emptyset$, we know
that $T_i^{i-1} = T_i = \emptyset$. By arc consistency, $T^{1}_i \neq
\emptyset$. Let $j \in \{2,\ldots,i-1\}$ be minimal such that $T_i^j
= \emptyset$. Thus
\begin{equation}
R_{ji}(s(x_j)) \cap T_i^{j-1} \ = \ \emptyset   \label{eq:RTj}
\end{equation}
By arc consistency, $\exists b \in D(x_i)$ such that $(s(x_j),b) \in
R_{ji}$ (i.e. $b \in R_{ji}(s(x_j))$). By Equation (\ref{eq:RTj}) and
definition of $T_i^{j-1}$, there is some $r < j$ such that
$(s(x_r),b) \notin R_{ri}$. Choose $r$ to be minimal. Consider $a \in
T_i^{j-1}$ (which is non-empty since $j$ was chosen to be minimal).
Then $(s(x_r),a) \in R_{ri}$, by definition of $T_i^{j-1}$, since $r
\leq j-1$. But $(s(x_j),a) \notin R_{ji}$ by (\ref{eq:RTj}). By arc
consistency, $\exists c \in D_r$ such that $(c,b) \in R_{ri}$ and
$\exists d \in D(x_j)$ such that $(d,a) \in R_{ji}$. By our inductive
hypothesis that $s(x_1),\ldots,s(x_{i-1})$ is a solution to
$I_{i-1}$, we know that $(s(x_r),s(x_j)) \in R_{rj}$ since $r,j < i$.
Because of their different compatibilities with, respectively, $b$
and $a$, we know that $c \neq s(x_r)$ and $d \neq s(x_j)$. If $c <
s(x_r)$, then, since $r<i$, the pattern BTI
occurs in $I$, as shown in Figure~\ref{fig:occurrencesbti}(a) on
page~\pageref{fig:occurrencesbti};
%
%
so we can deduce that $c > s(x_r)$. Similarly, if $d
< s(x_j)$, then, since $j<i$, the pattern BTI occurs in $I$, as shown in
Figure~\ref{fig:occurrencesbti}(b) on
page~\pageref{fig:occurrencesbti};
%
so we can deduce that $d > s(x_j)$.

But, by definition of $s$, $s(x_r)$ is the maximal element of $T_r$.
So, since $c > s(x_r)$, there must be some $t < r$ such that
$(s(x_t),c) \notin R_{tr}$. By minimality of $r$, we know that
$(s(x_t),b) \in R_{ti}$. Similarly, there must be some $k < j$ such
that $(s(x_k),d) \notin R_{kj}$ since $d > s(x_j)$ and $s(x_j)$ is
the maximal element of $T_j$. Since $a \in T_i^{j-1}$ and $k < j$, we
know that $(s(x_k),a) \in R_{ki}$.
 By our inductive hypothesis that $s(x_1),\ldots,s(x_{i-1})$ is
a solution to $I_{i-1}$, we know that $(s(x_k),s(x_j)) \in R_{kj}$
and $(s(x_t),s(x_r)) \in R_{tr}$ since $k,j,t,r < i$. We know that $a
\neq b$ because they have different compatibilities with $s(x_j)$. If
$a > b$, then, since $k < j$, the pattern BTI occurs in $I$, as
shown in Figure~\ref{fig:occurrencesbti}(c) on
page~\pageref{fig:occurrencesbti}.
%
And, if $b > a$, then the
pattern BTI occurs in $I$, as shown in Figure~\ref{fig:occurrencesbti}(d) on
page~\pageref{fig:occurrencesbti}.
%
%
This contradiction shows that $T_i \neq \emptyset$ (for each
$i=1,\ldots,n$) and hence that $I$ has a solution $s$.
\end{proof}

We conclude this section with a pattern which is essentially different from
the patterns EMC, BTX, and BTI, since it includes two negative edges that meet
but has no domain or variable order. The tractability of this pattern was
previously unknown~\cite{Escamocher14:thesis}.

\begin{figure}
\centering
\begin{picture}(120,100)(0,0)

\put(10,50){\usebox{\vartwo}} \put(50,10){\usebox{\varone}}
\put(90,50){\usebox{\vartwo}}
\dashline{5}(20,60)(60,30) \dashline{5}(60,30)(100,60)
\put(20,80){\line(4,-1){80}} \put(20,60){\line(4,1){80}}
\put(20,80){\line(4,-5){40}} \put(60,30){\line(4,5){40}}
\put(60,8){\makebox(0,0){$x$}}
\put(20,42){\makebox(0,0){$y$}}
\put(100,42){\makebox(0,0){$z$}}
\put(5,80){\makebox(0,0){$\alpha$}}
\put(5,60){\makebox(0,0){$\beta$}}
\put(115,80){\makebox(0,0){$\gamma$}}
\put(115,60){\makebox(0,0){$\delta$}}
\put(75,30){\makebox(0,0){$\epsilon$}}

\end{picture}
\caption{The pattern LX.}
\label{fig:lx}
\end{figure}
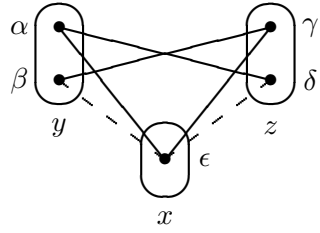

\begin{thm} \label{thm:lx}
AC is a decision procedure for CSP$_{\overline{SP}}(LX)$ where LX is
the pattern shown in Figure~\ref{fig:lx}. 
\end{thm}

\begin{proof}
Since establishing arc consistency only eliminates domain elements,
and hence cannot introduce the pattern, we only need to show that
every arc-consistent instance $I \in$
CSP$_{\overline{SP}}(LX)$ has a solution. In fact we will show a
stronger result by proving that the hypothesis $H_n$, below, holds
for all $n \geq 1$.
\begin{description}
\item[$H_n$]  \ for all arc-consistent instances $I=\tuple{X, D, A, \cpt}
\in$ CSP$_{\overline{SP}}(LX)$ with $|X| = n$,
$\forall x_i \in X$, $\forall a \in D(x_i)$, $I$ has a solution $s$ such that $s(x_i) = a$.
\end{description}
Trivially, $H_1$ holds. Suppose that $H_{n-1}$ holds where $n>1$. We
will show that this implies $H_n$, which will complete the proof by
induction.

Consider an arc-consistent instance $I=\tuple{X, D, A, \cpt}$ from CSP$_{\overline{SP}}(LX)$
with $X = \{x_1,\ldots,x_n\}$ and let $a
\in D(x_i)$ where $1 \leq i \leq n$. Let $I_{n-1}$ denote the
subproblem of $I$ on variables $X \setminus \{x_i\}$. For any
solution $s$ of $I_{n-1}$, we denote by $CV(\tuple{x_i,a},s)$ the set
of variables in $X \setminus \{x_i\}$ on which $s$ is compatible with
the unary assignment $\tuple{x_i,a}$, i.e.
$$CV(\tuple{x_i,a},s) \ \ = \ \ \{x_j \in X \setminus \{x_i\} \mid (a,s(x_j)) \in R_{ij} \}
$$
Consider two distinct solutions $s,s'$ to $I_{n-1}$. If we have $x_j
\in CV(\tuple{x_i,a},s) \setminus CV(\tuple{x_i,a},s')$ and $x_k \in
CV(\tuple{x_i,a},s') \setminus CV(\tuple{x_i,a},s)$, then the pattern LX
occurs in $I$ under the mapping $x\mapsto x_i$, $y\mapsto x_j$, $z\mapsto x_k$,
$\alpha\mapsto s(x_j)$,
$\beta\mapsto s'(x_j)$, $\gamma\mapsto s'(x_k)$, $\delta\mapsto s(x_k)$,
$\epsilon\mapsto a$ (see
Figure~\ref{fig:lx} on
page~\pageref{fig:lx}). Since LX does not occur in $I$, we can deduce
that the sets $CV(\tuple{x_i,a},s)$, as $s$ varies over all solutions
to $I_{n-1}$, form a nested family of sets. Let $s_a$ be a solution
to $I_{n-1}$ such that $CV(\tuple{x_i,a},s_a)$ is maximal for
inclusion.

Consider any $x_j \in X \setminus \{x_i\}$. By arc consistency,
$\exists b \in D(x_j)$ such that $(a,b) \in R_{ij}$. By our inductive
hypothesis $H_{n-1}$, there is a solution $s$ to $I_{n-1}$ such that
$s(x_j)=b$. Since $(a,s(x_j))=(a,b) \in R_{ij}$, we have $x_j \in
CV(\tuple{x_i,a},s)$. By maximality of $s_a$, this implies $x_j \in
CV(\tuple{x_i,a},s_a)$, i.e. $(a,s_a(x_j)) \in R_{ij}$. Since this is
true for any $x_j \in X \setminus \{x_i\}$, we can deduce that $s_a$
can be extended to a solution to $I$ (which assigns $a$ to $x_i$) by
simply adding the assignment $\tuple{x_i,a}$ to $s_a$.
\end{proof}

The proof of Theorem~\ref{thm:lx} shows that to find a solution to arc-consistent
instances in CSP$_{\overline{SP}}(LX)$, it suffices to maintain arc consistency during search (a 
standard algorithm in CSP solvers) making arbitrary choices of values to assign
to each variable: search will always be backtrack-free.

\section{Recognition problem for unknown orders}

One possible way of applying our new tractable classes is
the automatic recognition by a general-purpose CSP solver that an instance (or a sub-instance encountered
during search after instantiation of a subset of variables) belongs to one of our tractable classes.
In this section we study the complexity of this recognition problem. 

For an unordered pattern $P$ of size $k$, checking for (the non-occurrence of) $P$ in
a CSP instance $I$ is solvable in time $O(|I|^k)$ by simple exhaustive search.
Consequently, checking for (the non-occurrence of) unordered patterns of
constant size is solvable in polynomial time. However, the situation is less
obvious for ordered patterns since we have to test all possible orderings of $I$.

The following result was shown in~\cite{cjs10:aij-btp}.

\begin{thm}\label{thm:recog-btp}
Given a binary CSP instance $I$ with a fixed total order on the domain, there is
a polynomial-time algorithm to find a total variable ordering such that BTP does not
occur in $I$ (or to determine that no such ordering exists).
\end{thm}

We show that the same result holds for the other three ordered patterns studied
in this paper, namely BTI, BTX, and EMC.

\begin{thm}\label{thm:recog-bti}
Given a binary CSP instance $I$ with a fixed total order on the domain
and a pattern $P \in \{$BTI, BTX, EMC$\}$, there is
a polynomial-time algorithm to find a total variable ordering such that
$P$ does not occur in $I$ (or to determine that no such ordering exists).
\end{thm}

\begin{proof}
We give a proof only for BTX as the same idea works for the other two patterns as
well. Given a binary CSP instance $I$ with $n$ variables
$x_1,\ldots,x_n$, we define an associated CSP instance $\Pi_I$ that has
a solution precisely when there exists a suitable variable ordering for $I$.
To construct $\Pi_I$, let $O_1,\ldots,O_n$ be variables taking values in
$\{1,\ldots,n\}$ representing positions in the ordering. We impose the ternary constraint
$O_i>\min(O_j,O_k)$ for all triples of variables $x_i,x_j,x_k$ in $I$ such that the
BTX pattern occurs for some $\alpha,\beta\in D(x_i)$ with $\alpha>\beta$,
$\epsilon\in D(x_j)$, and $\gamma,\delta\in D(x_k)$ when the variables are ordered $x_i<x_j,x_k$.
The instance $\Pi_I$ has a solution precisely if there is an ordering of the
variables $x_1,\ldots,x_n$ of $I$ for which BTX does not occur.
Note that if the solution obtained represents a partial order (i.e. if
$O_i$ and $O_j$ are assigned the same value for some $i\neq j$), then it can be extended
to a total order which still satisfies all the constraints
by arbitrarily choosing the order of those $O_i$'s that are assigned the same value.
This reduction is polynomial in the size of $I$.
We now show that all constraints in $\Pi_I$ are ternary min-closed and thus $\Pi_I$ can be
solved in polynomial time~\cite{Jeavons95:maxclosed}.
Let $\tuple{p_1,q_1,r_1}$
and $\tuple{p_2,q_2,r_2}$ satisfy any constraint in $\Pi_I$. Then
$p_1>\min(q_1,r_1)$ and $p_2>\min(q_2,r_2)$, and thus
$\min(p_1,p_2)>\min(\min(q_1,r_1),\min(q_2,r_2))$ $=$
$\min(\min(q_1,q_2),\min(r_1,r_2))$.
Consequently,
$\langle\min(p_1,p_2)$, $\min(q_1,q_2)$, $\min(r_1,r_2)\rangle$ also satisfies the
constraint. We can deduce that all constraints in $\Pi_I$ are min-closed.
\end{proof}

Using the same technique, we can also show the following.

\begin{thm} \label{thm:recog-dom-bti}
Given a binary CSP instance $I$ with a fixed total variable order and
a pattern $P \in \{$BTI, BTX$\}$, there is a polynomial-time
algorithm to find a total domain ordering such that $P$ does not
occur in $I$ (or determine that no such ordering exists).
\end{thm}

It is known that determining a domain order for which MC does not occur is
NP-hard~\cite{Green08:ai}. Not surprisingly, for EMC when the domain order is
not known, detection becomes NP-hard. For the case of BTX and BTI, if neither
the domain nor variable order is known, finding orders for which the pattern
does not occur is again NP-hard.

\begin{thm} \label{thm:hardness}
For the pattern EMC, even for a fixed total variable order of an arc-consistent
binary CSP instance $I$, it is NP-hard to find
a total domain ordering of $I$ such that the pattern does not occur in $I$.
For patterns BTX and BTI, it is NP-hard to find
total variable and domain orderings of an arc-consistent
binary CSP instance $I$ such that the pattern does not occur in $I$.
\end{thm}

\begin{proof}
To show this, we exhibit a polynomial reduction from 3SAT. Given an
$n$-variable instance $I_{3SAT}$ of 3SAT with $m$ clauses, we create a basic domain $B$ of size
$2n$ with a value $a_i$ for each variable $X_i$ in $I_{3SAT}$ and
another $a_{n+i}$ for its negation $\overline{X_i}$. We construct a binary CSP instance $I_{CSP}$,
such that the domain of each variable contains $B$, and
such that there exists an appropriate order of $B$ if and only if $I_{3SAT}$
has a solution. For
each total ordering $>$ of $B$ there is a corresponding
assignment to the variables of $I_{3SAT}$ given by $X_i =$ {\tt true}
if and only if $a_i > a_{i+n}$. To complete the reduction we have to
show
how to impose a clause, e.g. $(a_i > a_{i+n}) \vee (a_j > a_{j+n})
\vee (a_k > a_{k+n})$. The basic construction of $I_{CSP}$ is
composed of the following elements: $N=6m$ variables linked by the equality
constraints $x_i = x_{i+1}$ ($i=1,\ldots,N-1$).
%

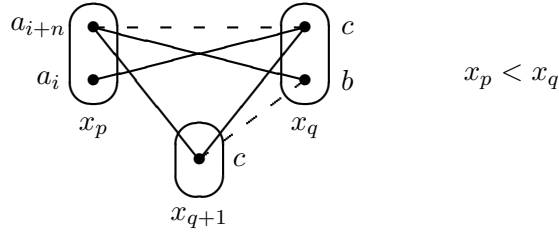
\begin{figure}
\centering
\begin{picture}(185,100)(-5,0)

\put(10,50){\usebox{\vartwo}} \put(50,10){\usebox{\varone}}
\put(90,50){\usebox{\vartwo}}
\dashline{5}(20,80)(100,80) \dashline{5}(60,30)(100,60)
\put(20,80){\line(4,-1){80}} \put(20,60){\line(4,1){80}}
\put(20,80){\line(4,-5){40}} \put(60,30){\line(4,5){40}}
\put(60,8){\makebox(0,0){$x_{q+1}$}}
  \put(20,42){\makebox(0,0){$x_{p}$}}  \put(100,42){\makebox(0,0){$x_q$}}
  \put(160,60){\makebox{$x_{p}<x_q$}}
  \put(-1,80){\makebox(0,0){$a_{i+n}$}}
 \put(3,60){\makebox(0,0){$a_i$}}
  \put(116,80){\makebox(0,0){$c$}}
  \put(116,60){\makebox(0,0){$b$}}
\put(75,30){\makebox(0,0){$c$}}

\end{picture}
\caption{To avoid the pattern EMC, we must have $(a_i > a_{i+n})$ or $ (b > c)$.}
\label{fig:avoidemc}
\end{figure}

Consider first the case of EMC. We assume the variable order given by
$x_i < x_j$ iff $i < j$. To impose a clause we construct a
gadget on three variables $x_p,x_q,x_r$ of $I_{CSP}$ which are not linked by
any other constraints in our construction (in particular, not
consecutive variables linked by equality constraints). To be specific, for the $i$th clause,
we choose $p=2i-1$, $q=4i-1$ and $r=6i-1$. We add four
extra values $a_{\max},b,c,d$ (of which $b,c,d$ depend on the clause)
to the basic domain $B$ of each of these three variables  $x_p,x_q,x_r$. The domain size
of each variable $x_i$ in $I_{CSP}$ is thus $2n+4$. The
value $a_{\max}$ is compatible with all values in the domains of the
other variables. This ensures arc consistency for all domain values,
and if $a_{\max} > a$ for all domain values $a \neq a_{\max}$, then
the pattern cannot occur on $a_{\max}$. We now specify the constraints
between the variables $x_p,x_q,x_r$. We place negative edges
between all pairs of values other than $a_{\max}$ except for three
positive edges in each constraint. In the constraint between $x_p$
and $x_q$ we add the three positive edges: $(a_i,b)$, $(a_i,c)$,
$(a_{i+n},b)$. If $a_{i+n} > a_i$ in $D(x_p)$ and $c > b$ in
$D(x_q)$, then the pattern EMC occurs in $I_{CSP}$ as shown in Figure~\ref{fig:avoidemc}. 
The third variable of the pattern is $x_{q+1}$ 
which is linked by
an equality constraint to $x_q$ and by a trivial constraint to $x_p$. Thus, to avoid the pattern occurring
in $I_{CSP}$ the domain order must respect
$$(a_i > a_{i+n}) \vee (b > c)$$
By adding similar constraints between $x_p,x_r$ and $x_q,x_r$, we can
also impose
\begin{eqnarray*}
& & (a_j > a_{j+n}) \vee (c > d) \ \ \ \ {\rm and} \\
& & (a_k > a_{k+n}) \vee (d > b)
\end{eqnarray*}
By imposing these three inequalities, we impose $(a_i > a_{i+n}) \vee
(a_j > a_{j+n}) \vee (a_k > a_{k+n})$ (since we cannot simultaneously
have $b>c>d>b$) which corresponds to the clause $X_i \vee X_j \vee
X_k$ in $I_{3SAT}$. By inversing the roles of $a_i$ and $a_{i+n}$ we
can clearly impose clauses involving negative literals. This
completes the reduction from $I_{3SAT}$ to the problem of finding a
domain ordering of a binary CSP instance so that EMC does not occur.
Since this reduction is clearly polynomial, we can conclude that the
problem of testing the existence of a domain order so that EMC does
not occur is NP-hard.

Now consider the case of the pattern BTX. We use a similar
construction to the case of EMC, above. Again, to simulate a clause
$X_i \vee X_j \vee X_k$ in $I_{3SAT}$ we need to impose $(a_i >
a_{i+n}) \vee (a_j > a_{j+n}) \vee (a_k > a_{k+n})$. This can be
achieved by imposing:
\begin{eqnarray*}
& & (a_i > a_{i+n}) \vee (x_p > x_q), \\
& & (a_j > a_{j+n}) \vee (x_q > x_r)  \ \ \ \ {\rm and} \\
& & (a_k > a_{k+n}) \vee (x_r > x_p)
\end{eqnarray*}
For example, to impose $(a_i > a_{i+n}) \vee (x_p > x_q)$ we place
the same gadget as above (i.e. positive edges $(a_i,b)$, $(a_i,c)$,
$(a_{i+n},b)$) on each of the three pairs of variables $(x_p,x_q)$,
$(x_{p+1},x_q)$ and $(x_{q+1}, x_p)$. To avoid BTX on variables
$x_p,x_q,x_{q+1}$, on variables  $x_{p+1},x_q,x_{q+1}$ and on
variables $x_{q+1},x_p,x_{p+1}$, we must have
\begin{eqnarray*}
& & (a_i > a_{i+n}) \vee (x_p > x_q) \vee (x_p > x_{q+1}), \\
& & (a_i > a_{i+n}) \vee (x_{p+1} > x_q) \vee (x_{p+1} > x_{q+1})  \ \ \ \ {\rm and} \\
& & (a_i > a_{i+n}) \vee (x_{q+1} > x_p) \vee (x_{q+1} > x_{p+1})
\end{eqnarray*}
Since there is a total strict ordering on the variables, this is
logically equivalent to imposing $(a_i > a_{i+n}) \vee (x_p > x_q)$,
as required (provided none of the variables $x_p,x_{p+1},x_q,x_{q+1}$
are used in other gadgets), given our freedom to choose the positions of
$x_{p+1}$ and $x_{q+1}$ in the order.

Finally, we consider the pattern BTI. But this is an easier case than
BTX. We just need to place the gadget on $(x_p,x_q)$ to impose $(a_i
> a_{i+n}) \vee (x_p > x_q)$ to avoid the pattern BTI.

Thus, EMC is NP-hard to detect when the domain order of the instance
is not fixed, and BTX, BTI are NP-hard to detect when neither the
domain order nor the variable order of the instance is fixed.
\end{proof}

The results from this section are summarised in Table~\ref{tab:summary}. We use
the star to denote uninteresting cases. Note that since LX is an unordered
pattern the questions of determining variable and/or domain orders are not
interesting. Similarly, since pattern BTP only orders variables the question of
determining a domain order is not interesting. An important point is that for a fixed domain order
(which is a natural assumption for numerical domain values, for example) we can effectively
exhaust over all variable orders in polynomial time for all variable-ordered patterns
studied in this paper (BTP, BTI, BTX, EMC).

\begin{table}[h]
\begin{center}
\begin{tabular}{cccccc}
& \textbf{BTP} & \textbf{BTI} & \textbf{BTX} & \textbf{EMC} & \textbf{LX} \\
\hline
domain order given & P {\small [Thm~\ref{thm:recog-btp}]} & P {\small [Thm~\ref{thm:recog-bti}]} & P {\small [Thm~\ref{thm:recog-bti}]} & P {\small [Thm~\ref{thm:recog-bti}]} & * \\
variable order given & * & P {\small [Thm~\ref{thm:recog-dom-bti}]} & P {\small [Thm~\ref{thm:recog-dom-bti}]}& NP-h {\small [Thm~\ref{thm:hardness}]} & * \\
neither order given & P {\small [Thm~\ref{thm:recog-btp}]} & NP-h {\small
[Thm~\ref{thm:hardness}]} & NP-h {\small [Thm\ref{thm:hardness}]} & NP-h {\small
[Thm~\ref{thm:hardness}]} & *
\end{tabular}
\end{center}
\caption{Summary of the complexity (P or NP-hard) of the recognition problem for each of our patterns
for the cases in which the domain order is given (but the variable order has to be determined), 
the variable order is given (but the domain order has to be determined), or neither order is given.}
\label{tab:summary}
\end{table}

\section{Characterisation of patterns solved by AC}

\subsection{Instances not solved by arc consistency} \label{sec:5.1}

We first give a set of instances, each of which is arc consistent and
has no solution. If for any of these instances $I$, we have $I \in$
CSP$_{\overline{SP}}(P)$, then this constitutes a proof, by
Lemma~\ref{lem:not-ac}, that pattern
$P$ is not solved by arc consistency. For simplicity of presentation,
in each of the following instances,
we suppose the variable order given by $x_i < x_j$ if $i < j$.

\thicklines  \setlength{\unitlength}{1.5pt}
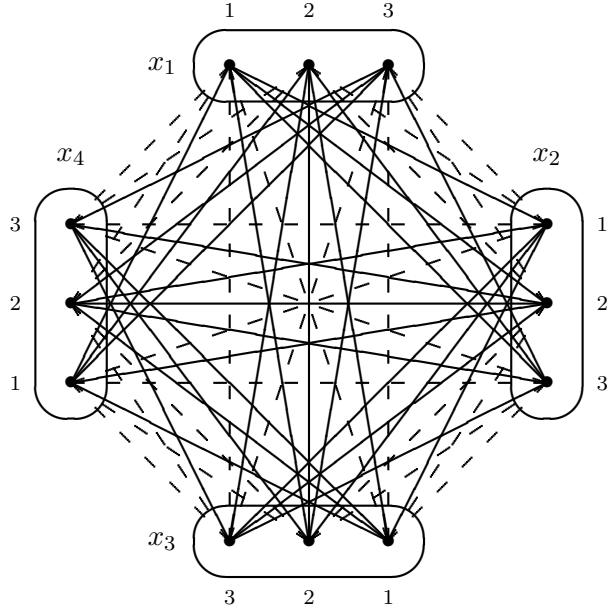
\begin{figure}
\centering
\begin{picture}(160,160)(10,10)

\put(30,90){\oval(18,58)} \put(30,70){\makebox(0,0){$\bullet$}}
\put(30,90){\makebox(0,0){$\bullet$}} \put(30,110){\makebox(0,0){$\bullet$}}
\put(150,90){\oval(18,58)} \put(150,70){\makebox(0,0){$\bullet$}}
\put(150,90){\makebox(0,0){$\bullet$}} \put(150,110){\makebox(0,0){$\bullet$}}
\put(90,30){\oval(58,18)} \put(70,30){\makebox(0,0){$\bullet$}}
\put(90,30){\makebox(0,0){$\bullet$}} \put(110,30){\makebox(0,0){$\bullet$}}
\put(90,150){\oval(58,18)} \put(70,150){\makebox(0,0){$\bullet$}}
\put(90,150){\makebox(0,0){$\bullet$}} \put(110,150){\makebox(0,0){$\bullet$}}

\put(53,150){\makebox(0,0){$x_1$}}
\put(150,127){\makebox(0,0){$x_2$}}
\put(53,30){\makebox(0,0){$x_3$}}
\put(30,127){\makebox(0,0){$x_4$}}

\dashline{4}(30,110)(70,150) \dashline{4}(30,110)(90,150)
\dashline{4}(30,90)(70,150) \dashline{4}(30,90)(90,150)
\dashline{4}(30,70)(150,70) \dashline{4}(30,70)(150,110)
\dashline{4}(30,110)(150,70) \dashline{4}(30,110)(150,110)
\dashline{4}(70,30)(70,150) \dashline{4}(70,30)(110,150)
\dashline{4}(110,30)(70,150) \dashline{4}(110,30)(110,150)
\dashline{4}(90,150)(150,90) \dashline{4}(90,150)(150,110)
\dashline{4}(110,150)(150,90) \dashline{4}(110,150)(150,110)
\dashline{4}(30,70)(70,30) \dashline{4}(30,70)(90,30)
\dashline{4}(30,90)(70,30) \dashline{4}(30,90)(90,30)
\dashline{4}(90,30)(150,90) \dashline{4}(90,30)(150,70)
\dashline{4}(110,30)(150,90) \dashline{4}(110,30)(150,70)

\put(30,90){\line(1,0){120}} \put(30,90){\line(6,1){120}} \put(30,90){\line(6,-1){120}}
\put(150,90){\line(-6,1){120}} \put(150,90){\line(-6,-1){120}}
\put(90,30){\line(0,1){120}} \put(90,30){\line(1,6){20}} \put(90,30){\line(-1,6){20}}
\put(90,150){\line(1,-6){20}} \put(90,150){\line(-1,-6){20}}
\put(30,70){\line(1,2){40}} \put(30,70){\line(3,4){60}} \put(30,70){\line(1,1){80}}
\put(110,150){\line(-2,-1){80}} \put(110,150){\line(-4,-3){80}}
\put(30,110){\line(1,-2){40}} \put(30,110){\line(3,-4){60}} \put(30,110){\line(1,-1){80}}
\put(110,30){\line(-2,1){80}} \put(110,30){\line(-4,3){80}}
\put(70,150){\line(2,-1){80}} \put(70,150){\line(4,-3){80}} \put(70,150){\line(1,-1){80}}
\put(150,70){\line(-1,2){40}} \put(150,70){\line(-3,4){60}}
\put(70,30){\line(2,1){80}} \put(70,30){\line(4,3){80}} \put(70,30){\line(1,1){80}}
\put(150,110){\line(-1,-2){40}} \put(150,110){\line(-3,-4){60}}

\put(70,164){\makebox(0,0){\scriptsize 1}} \put(90,164){\makebox(0,0){\scriptsize 2}}
 \put(110,164){\makebox(0,0){\scriptsize 3}}
\put(70,16){\makebox(0,0){\scriptsize 3}} \put(90,16){\makebox(0,0){\scriptsize 2}}
 \put(110,16){\makebox(0,0){\scriptsize 1}}
\put(16,70){\makebox(0,0){\scriptsize 1}} \put(16,90){\makebox(0,0){\scriptsize 2}}
 \put(16,110){\makebox(0,0){\scriptsize 3}}
\put(164,70){\makebox(0,0){\scriptsize 3}} \put(164,90){\makebox(0,0){\scriptsize 2}}
 \put(164,110){\makebox(0,0){\scriptsize 1}}

\end{picture}
\caption{The instance $I_{K4}$.}
\label{fig:IK4}
\end{figure}

\thicklines  \setlength{\unitlength}{1.4pt}
\begin{figure}
\centering
\begin{picture}(170,160)(10,10)

\put(30,90){\oval(18,58)} \put(30,70){\makebox(0,0){$\bullet$}}
\put(30,90){\makebox(0,0){$\bullet$}}
\put(30,110){\makebox(0,0){$\bullet$}} \put(150,100){\oval(18,38)}
\put(150,90){\makebox(0,0){$\bullet$}}
\put(150,110){\makebox(0,0){$\bullet$}} \put(80,30){\oval(38,18)}
\put(70,30){\makebox(0,0){$\bullet$}}
\put(90,30){\makebox(0,0){$\bullet$}} \put(80,150){\oval(38,18)}
\put(70,150){\makebox(0,0){$\bullet$}}
\put(90,150){\makebox(0,0){$\bullet$}}

\put(55,150){\makebox(0,0){$x_1$}}
\put(150,125){\makebox(0,0){$x_3$}} \put(55,30){\makebox(0,0){$x_2$}}
\put(30,125){\makebox(0,0){$x_0$}}

\dashline{4}(30,110)(70,150) \dashline{4}(30,90)(70,150)
\dashline{4}(30,70)(70,30) \dashline{4}(30,110)(70,30)
\dashline{4}(30,70)(150,110) \dashline{4}(30,90)(150,110)
\dashline{4}(90,150)(150,90) \dashline{4}(90,30)(150,90)
\dashline{4}(90,30)(90,150)

\put(30,90){\line(1,0){120}} \put(30,110){\line(1,0){120}}
\put(150,90){\line(-6,1){120}} \put(150,90){\line(-6,-1){120}}

\put(30,90){\line(1,1){60}} \put(30,70){\line(1,2){40}}
\put(30,110){\line(3,2){60}} \put(30,70){\line(3,4){60}}

\put(30,90){\line(1,-1){60}} \put(30,90){\line(2,-3){40}}
\put(30,110){\line(3,-4){60}} \put(30,70){\line(3,-2){60}}

\put(70,30){\line(0,1){120}} \put(70,30){\line(1,6){20}}
\put(90,30){\line(-1,6){20}} \put(70,150){\line(2,-1){80}}
\put(70,150){\line(4,-3){80}} \put(90,150){\line(3,-2){60}}
\put(70,30){\line(1,1){80}} \put(70,30){\line(4,3){80}}
\put(90,30){\line(3,4){60}}

\put(70,164){\makebox(0,0){\scriptsize 1}}
\put(90,164){\makebox(0,0){\scriptsize 0}}
\put(70,16){\makebox(0,0){\scriptsize 1}}
\put(90,16){\makebox(0,0){\scriptsize 0}}
\put(16,70){\makebox(0,0){\scriptsize 1}}
\put(16,90){\makebox(0,0){\scriptsize 2}}
\put(16,110){\makebox(0,0){\scriptsize 3}}
\put(164,90){\makebox(0,0){\scriptsize 0}}
\put(164,110){\makebox(0,0){\scriptsize 1}}


\end{picture}
\caption{The instance $I_{4}$.} \label{fig:I4}
\end{figure}
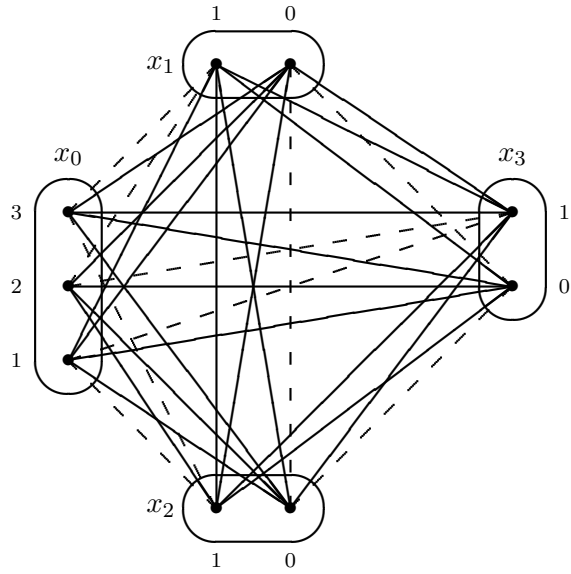
\thicklines \setlength{\unitlength}{1pt}

\thicklines  \setlength{\unitlength}{1.3pt}%
\newsavebox{\vart}
\savebox{\vart}(20,40){
\begin{picture}(20,40)(0,0)
\put(10,20){\oval(18,38)} \put(10,10){\makebox(0,0){$\bullet$}}
\put(10,30){\makebox(0,0){$\bullet$}}
\end{picture}
}
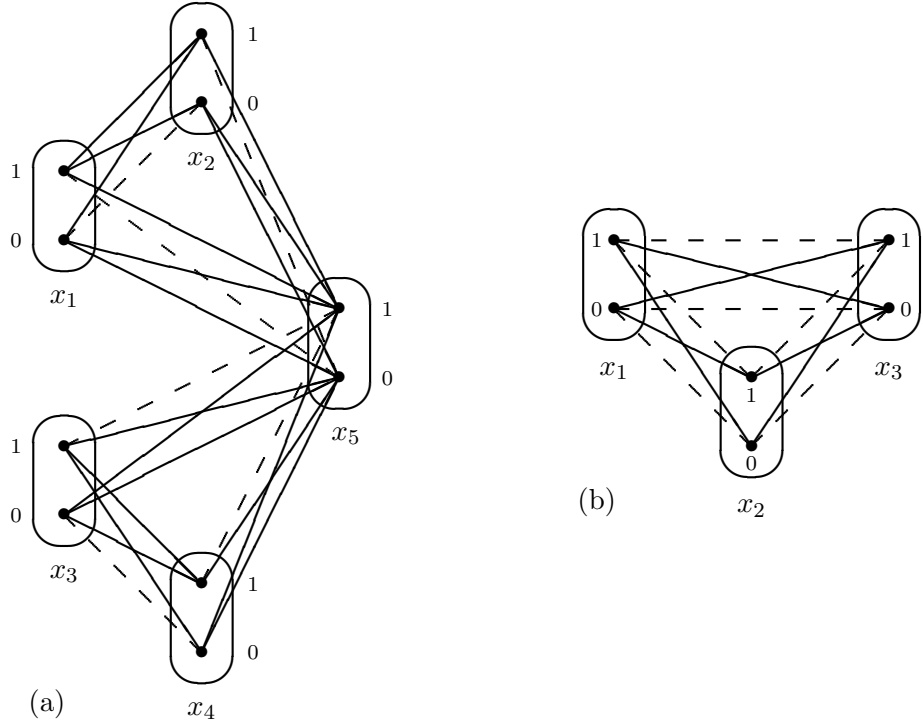
\begin{figure}
\centering
\begin{picture}(280,220)(0,0)

\put(0,0){
\begin{picture}(120,220)(0,-20)
\put(10,30){\usebox{\vart}} \put(50,-10){\usebox{\vart}}
\put(90,70){\usebox{\vart}} \put(10,110){\usebox{\vart}}
\put(50,150){\usebox{\vart}}
\dashline{5}(20,120)(60,160) \dashline{5}(20,140)(100,80) \dashline{5}(60,180)(100,80)
\dashline{5}(20,60)(100,100) \dashline{5}(20,40)(60,0) \dashline{5}(60,20)(100,100)
\put(20,140){\line(1,1){40}} \put(20,140){\line(2,1){40}} \put(20,120){\line(2,3){40}}
\put(20,140){\line(2,-1){80}} \put(20,120){\line(2,-1){80}} \put(20,120){\line(4,-1){80}}
\put(60,180){\line(1,-2){40}} \put(60,160){\line(1,-2){40}} \put(60,160){\line(2,-3){40}}
\put(20,60){\line(4,1){80}} \put(20,40){\line(4,3){80}} \put(20,40){\line(2,1){80}}
\put(20,60){\line(2,-3){40}} \put(20,60){\line(1,-1){40}} \put(20,40){\line(2,-1){40}}
\put(60,20){\line(2,3){40}} \put(60,0){\line(1,2){40}} \put(60,0){\line(2,5){40}}
\put(6,140){\makebox(0,0){\scriptsize 1}}
\put(6,120){\makebox(0,0){\scriptsize 0}}
\put(6,60){\makebox(0,0){\scriptsize 1}}
\put(6,40){\makebox(0,0){\scriptsize 0}}
\put(75,180){\makebox(0,0){\scriptsize 1}}
\put(75,160){\makebox(0,0){\scriptsize 0}}
\put(75,20){\makebox(0,0){\scriptsize 1}}
\put(75,0){\makebox(0,0){\scriptsize 0}}
\put(114,100){\makebox(0,0){\scriptsize 1}}
\put(114,80){\makebox(0,0){\scriptsize 0}}
\put(20,103){\makebox(0,0){$x_1$}} \put(60,143){\makebox(0,0){$x_2$}}  \put(20,23){\makebox(0,0){$x_3$}}
\put(60,-17){\makebox(0,0){$x_4$}} \put(102,63){\makebox(0,0){$x_5$}}

\put(15,-15){\makebox(0,0){(a)}}
\end{picture}
}

\put(160,60){
\begin{picture}(280,100)(0,0)
\put(10,50){\usebox{\vart}} \put(50,10){\usebox{\vart}} \put(90,50){\usebox{\vart}}
\dashline{5}(20,60)(100,60) \dashline{5}(20,80)(60,40) \dashline{5}(60,20)(100,60)
\dashline{5}(20,80)(100,80) \dashline{5}(20,60)(60,20) \dashline{5}(60,40)(100,80)

\put(20,80){\line(4,-1){80}} \put(20,60){\line(4,1){80}}
\put(20,80){\line(2,-3){40}} \put(20,60){\line(2,-1){40}}
\put(60,40){\line(2,1){40}} \put(60,20){\line(2,3){40}}

\put(60,2){\makebox(0,0){$x_2$}}
\put(20,42){\makebox(0,0){$x_1$}}  \put(100,42){\makebox(0,0){$x_3$}}

\put(15,80){\makebox(0,0){\scriptsize 1}}
\put(15,60){\makebox(0,0){\scriptsize 0}}
\put(60,35){\makebox(0,0){\scriptsize 1}}
\put(60,15){\makebox(0,0){\scriptsize 0}}
\put(105,80){\makebox(0,0){\scriptsize 1}}
\put(105,60){\makebox(0,0){\scriptsize 0}}

\put(15,4){\makebox(0,0){(b)}}
\end{picture}
}

\end{picture}
\caption{The instances (a) $I^{SAT}_{2\Delta}$ and (b) $I^{2COL}_{3}$.} \label{fig:ISAT2Delta}
\end{figure}
\thicklines \setlength{\unitlength}{1pt}

\thicklines \setlength{\unitlength}{1.2pt}
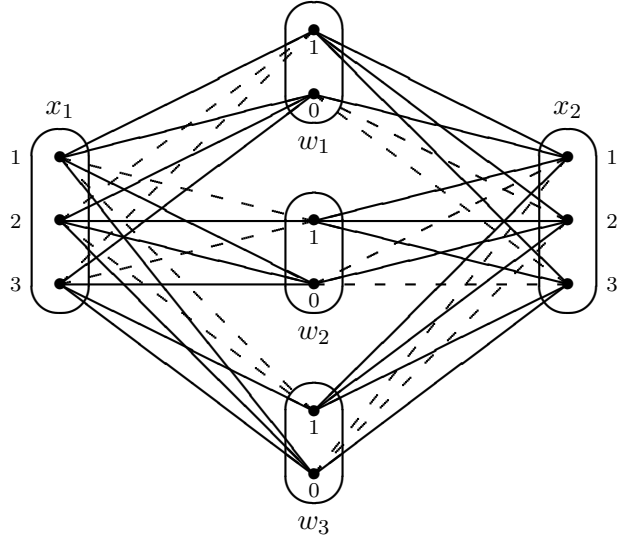
\begin{figure}
\centering
\begin{picture}(200,170)(10,10)

\put(30,110){\oval(18,58)} \put(30,90){\makebox(0,0){$\bullet$}}
\put(30,110){\makebox(0,0){$\bullet$}}
\put(30,130){\makebox(0,0){$\bullet$}} \put(110,40){\oval(18,38)}
\put(110,30){\makebox(0,0){$\bullet$}}
\put(110,50){\makebox(0,0){$\bullet$}} \put(110,100){\oval(18,38)}
\put(110,90){\makebox(0,0){$\bullet$}}
\put(110,110){\makebox(0,0){$\bullet$}} \put(110,160){\oval(18,38)}
\put(110,150){\makebox(0,0){$\bullet$}}
\put(110,170){\makebox(0,0){$\bullet$}} \put(190,110){\oval(18,58)}
\put(190,90){\makebox(0,0){$\bullet$}}
\put(190,110){\makebox(0,0){$\bullet$}}
\put(190,130){\makebox(0,0){$\bullet$}}

\put(16,90){\makebox(0,0){\scriptsize 3}}
\put(16,110){\makebox(0,0){\scriptsize 2}}
 \put(16,130){\makebox(0,0){\scriptsize 1}}
\put(204,90){\makebox(0,0){\scriptsize 3}}
\put(204,110){\makebox(0,0){\scriptsize 2}}
 \put(204,130){\makebox(0,0){\scriptsize 1}}
\put(110,165){\makebox(0,0){\scriptsize 1}}
\put(110,145){\makebox(0,0){\scriptsize 0}}
\put(110,105){\makebox(0,0){\scriptsize 1}}
\put(110,85){\makebox(0,0){\scriptsize 0}}
\put(110,45){\makebox(0,0){\scriptsize 1}}
\put(110,25){\makebox(0,0){\scriptsize 0}}

\put(30,145){\makebox(0,0){$x_1$}}
\put(190,145){\makebox(0,0){$x_2$}}
\put(110,134){\makebox(0,0){$w_1$}}
\put(110,74){\makebox(0,0){$w_2$}} \put(110,14){\makebox(0,0){$w_3$}}

\dashline{4}(30,90)(110,170) \dashline{4}(30,110)(110,170)
\dashline{4}(30,90)(110,110) \dashline{4}(30,130)(110,110)
\dashline{4}(30,110)(110,50) \dashline{4}(30,130)(110,50)
\dashline{4}(190,110)(110,150) \dashline{4}(190,90)(110,150)
\dashline{4}(190,130)(110,90) \dashline{4}(190,90)(110,90)
\dashline{4}(190,110)(110,30) \dashline{4}(190,130)(110,30)

\put(30,130){\line(2,1){80}} \put(30,130){\line(4,1){80}}
\put(30,130){\line(2,-1){80}} \put(30,130){\line(4,-5){80}}
\put(30,110){\line(2,1){80}} \put(30,110){\line(1,0){80}}
\put(30,110){\line(4,-1){80}} \put(30,110){\line(1,-1){80}}
\put(30,90){\line(4,3){80}} \put(30,90){\line(1,0){80}}
\put(30,90){\line(2,-1){80}} \put(30,90){\line(4,-3){80}}

\put(190,130){\line(-2,1){80}} \put(190,130){\line(-4,1){80}}
\put(190,130){\line(-4,-1){80}} \put(190,130){\line(-1,-1){80}}
\put(190,110){\line(-4,3){80}} \put(190,110){\line(-1,0){80}}
\put(190,110){\line(-4,-1){80}} \put(190,110){\line(-4,-3){80}}
\put(190,90){\line(-1,1){80}} \put(190,90){\line(-4,1){80}}
\put(190,90){\line(-2,-1){80}} \put(190,90){\line(-4,-3){80}}

\end{picture}
\caption{The instance $I_{5}$ (with variable order $w_1 < w_2 < w_3 < x_1 <
x_2$).} \label{fig:I5}
\end{figure}

\thicklines  \setlength{\unitlength}{1.3pt}%
\newsavebox{\vartwobig}
\savebox{\vartwobig}(20,40){
\begin{picture}(20,40)(0,0)
\put(10,20){\oval(18,38)} \put(10,10){\makebox(0,0){$\bullet$}}
\put(10,30){\makebox(0,0){$\bullet$}}
\end{picture}
}
\begin{figure}
\centering
\begin{picture}(280,100)(0,0)
\put(10,50){\usebox{\vartwobig}} \put(50,10){\usebox{\vartwobig}}
\put(90,50){\usebox{\vartwobig}} \put(170,50){\usebox{\vartwobig}}
\put(210,10){\usebox{\vartwobig}} \put(250,50){\usebox{\vartwobig}}
\dashline{5}(20,60)(100,60) \dashline{5}(20,80)(60,40) \dashline{5}(60,20)(100,60)
\dashline{5}(100,80)(180,80)
\dashline{5}(180,60)(260,60) \dashline{5}(180,60)(220,20) \dashline{5}(220,40)(260,80)
\put(20,80){\line(4,-1){80}} \put(20,60){\line(4,1){80}} \put(20,80){\line(1,0){80}}
\put(20,80){\line(2,-3){40}} \put(20,60){\line(2,-1){40}} \put(20,60){\line(1,-1){40}}
\put(60,40){\line(1,1){40}} \put(60,40){\line(2,1){40}} \put(60,20){\line(2,3){40}}
\put(180,80){\line(2,-3){40}} \put(180,80){\line(1,-1){40}} \put(180,60){\line(2,-1){40}}
\put(220,40){\line(2,1){40}} \put(220,20){\line(1,1){40}} \put(220,20){\line(2,3){40}}
\put(100,80){\line(4,-1){80}} \put(100,60){\line(4,1){80}} \put(100,60){\line(1,0){80}}
\put(180,80){\line(4,-1){80}} \put(180,60){\line(4,1){80}} \put(180,80){\line(1,0){80}}
\put(60,2){\makebox(0,0){$x_2$}}
\put(20,42){\makebox(0,0){$x_1$}}  \put(100,42){\makebox(0,0){$x_3$}}
\put(220,2){\makebox(0,0){$x_5$}}
\put(180,42){\makebox(0,0){$x_4$}}  \put(260,42){\makebox(0,0){$x_6$}}

\put(15,80){\makebox(0,0){\scriptsize 1}}
\put(15,60){\makebox(0,0){\scriptsize 0}}
\put(60,35){\makebox(0,0){\scriptsize 1}}
\put(60,15){\makebox(0,0){\scriptsize 0}}
\put(102,75){\makebox(0,0){\scriptsize 1}}
\put(102,55){\makebox(0,0){\scriptsize 0}}
\put(178,75){\makebox(0,0){\scriptsize 1}}
\put(178,55){\makebox(0,0){\scriptsize 0}}
\put(220,35){\makebox(0,0){\scriptsize 1}}
\put(220,15){\makebox(0,0){\scriptsize 0}}
\put(265,80){\makebox(0,0){\scriptsize 1}}
\put(265,60){\makebox(0,0){\scriptsize 0}}
\end{picture}
\caption{The instance $I^{SAT}_{6}$.} \label{fig:ISAT6}
\end{figure}
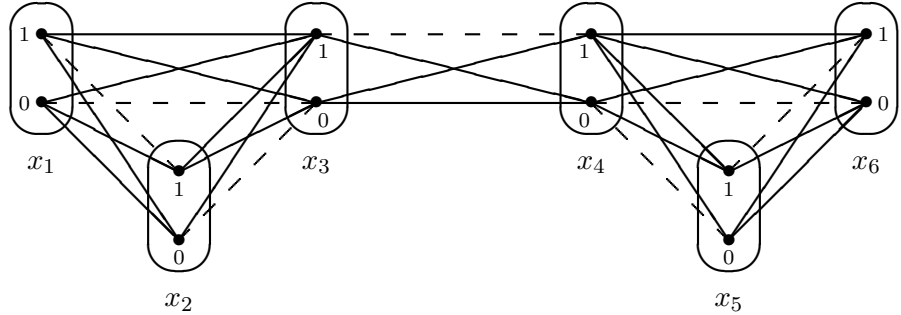
\thicklines \setlength{\unitlength}{1pt}

\thicklines  \setlength{\unitlength}{1.5pt}%
\begin{figure}
\centering
\begin{picture}(160,120)(0,0)   

\put(20,90){\oval(18,38)} \put(20,80){\makebox(0,0){$\bullet$}}
\put(20,100){\makebox(0,0){$\bullet$}} \put(40,30){\oval(18,38)}
\put(40,20){\makebox(0,0){$\bullet$}}
\put(40,40){\makebox(0,0){$\bullet$}} \put(120,30){\oval(18,38)}
\put(120,20){\makebox(0,0){$\bullet$}}
\put(120,40){\makebox(0,0){$\bullet$}} \put(140,90){\oval(18,38)}
\put(140,80){\makebox(0,0){$\bullet$}}
\put(140,100){\makebox(0,0){$\bullet$}}

\put(16,100){\makebox(0,0){\scriptsize 1}}
\put(16,80){\makebox(0,0){\scriptsize 0}}
\put(40,35){\makebox(0,0){\scriptsize 1}}
\put(40,15){\makebox(0,0){\scriptsize 0}}
\put(120,35){\makebox(0,0){\scriptsize 1}}
\put(120,16){\makebox(0,0){\scriptsize 0}}
\put(144,100){\makebox(0,0){\scriptsize 1}}
\put(144,80){\makebox(0,0){\scriptsize 0}}

\put(10,65){\makebox(0,0){$x_1$}} \put(40,5){\makebox(0,0){$x_2$}}
\put(120,5){\makebox(0,0){$x_3$}} \put(150,65){\makebox(0,0){$x_4$}}

\dashline{4}(20,80)(140,100) \dashline{4}(20,80)(120,40)
\dashline{4}(20,100)(40,40) \dashline{4}(40,20)(140,100)
\dashline{4}(40,20)(120,40) \dashline{4}(120,20)(140,80)

\put(20,100){\line(6,-1){120}} \put(20,100){\line(1,0){120}}
\put(20,100){\line(5,-3){100}} \put(20,100){\line(5,-4){100}}
\put(20,100){\line(1,-4){20}} \put(20,80){\line(1,0){120}}
\put(20,80){\line(5,-3){100}} \put(20,80){\line(1,-2){20}}
\put(20,80){\line(1,-3){20}} \put(40,40){\line(5,3){100}}
\put(40,40){\line(5,2){100}} \put(40,40){\line(1,0){80}}
\put(40,40){\line(4,-1){80}} \put(40,20){\line(5,3){100}}
\put(40,20){\line(1,0){80}} \put(120,40){\line(1,3){20}}
\put(120,40){\line(1,2){20}} \put(120,20){\line(1,4){20}}

\end{picture}
\caption{The instance $I^{SAT}_{K4}$.} \label{fig:ISATK4}
\end{figure}
\thicklines \setlength{\unitlength}{1pt}

\begin{itemize}
\item $I_{K4}$ (shown in Figure~\ref{fig:IK4}) 
is composed of four variables with domains $D(x_i) = \{1,2,3\}$ ($i=1,2,3,4$),
and the following constraints: $(x_i = 1) \vee (x_j = 3)$
($(i,j) = (1,2), (2,3), (3,4), (4,1)$) and
$(x_i = 2) \vee (x_j = 2)$ ($(i,j) = (1,3), (2,4)$).

\item $I_4$ (shown in Figure~\ref{fig:I4}) 
is composed of four variables with domains $D(x_0) = \{1,2,3\}$,
$D(x_i) = \{0,1\}$ ($i=1,2,3$),
and the following constraints: $x_i \vee x_j$ ($1 \leq i < j \leq 3$) and
$(x_0 = i) \vee \overline{x_i}$ ($i = 1,2,3$).

\item $I^{SAT}_{2\Delta}$ (shown in Figure~\ref{fig:ISAT2Delta}(a)) 
is composed of five Boolean variables and the following constraints:
$x_1 \vee x_2$, $x_3 \vee x_4$, $\overline{x_1} \vee x_5$, $\overline{x_2} \vee x_5$,
$\overline{x_3} \vee \overline{x_5}$, $\overline{x_4} \vee \overline{x_5}$.

\item $I_5$ (shown in Figure~\ref{fig:I5}) 
is composed of five variables with domains $D(w_i) = \{0,1\}$ ($i=1,2,3$),
$D(x_i) = \{1,2,3\}$, and the constraints: $\overline{w_i} \vee (x_1 =
i)$ ($i=1,2,3$) and $w_i \vee (x_2 = i)$ ($i=1,2,3$). In this
instance the variable order is $w_1 < w_2 < w_3 < x_1 < x_2$.

\item $I^{SAT}_6$ (shown in Figure~\ref{fig:ISAT6}) 
is composed of six Boolean variables and the following constraints:
$\overline{x_1} \vee \overline{x_2}$,  $x_1 \vee x_3$, $x_2 \vee x_3$, $\overline{x_3} \vee \overline{x_4}$,
$x_4 \vee x_5$, $x_4 \vee x_6$, $\overline{x_5} \vee \overline{x_6}$.

\item $I^{SAT}_{K4}$ (shown in Figure~\ref{fig:ISATK4}) 
is composed of four Boolean variables and the following constraints:
$\overline{x_1} \vee \overline{x_2}$, $x_3 \vee x_4$ and $x_i \vee
\overline{x_j}$ (for $(i,j) = (1,3)$, $(1,4)$, $(2,3)$, $(2,4)$).

\item $I^{2COL}_{3}$  (shown in Figure~\ref{fig:ISAT2Delta}(b)) 
is composed of three Boolean variables and the three
inequality constraints: $x_i \neq x_j$ ($1 \leq i < j \leq 3$).
\end{itemize}

\thicklines \setlength{\unitlength}{0.9pt}
\newsavebox{\varitwo}
\savebox{\varitwo}(20,40){
\begin{picture}(20,40)(0,0)
\put(10,20){\oval(18,38)} \put(10,10){\makebox(0,0){$\bullet$}}
\put(10,30){\makebox(0,0){$\bullet$}}
\end{picture}
}
\newsavebox{\varione}
\savebox{\varione}(20,40){
\begin{picture}(20,40)(0,0)
\put(10,20){\oval(18,28)} \put(10,20){\makebox(0,0){$\bullet$}}
\end{picture}
} 

\thicklines \setlength{\unitlength}{0.9pt}
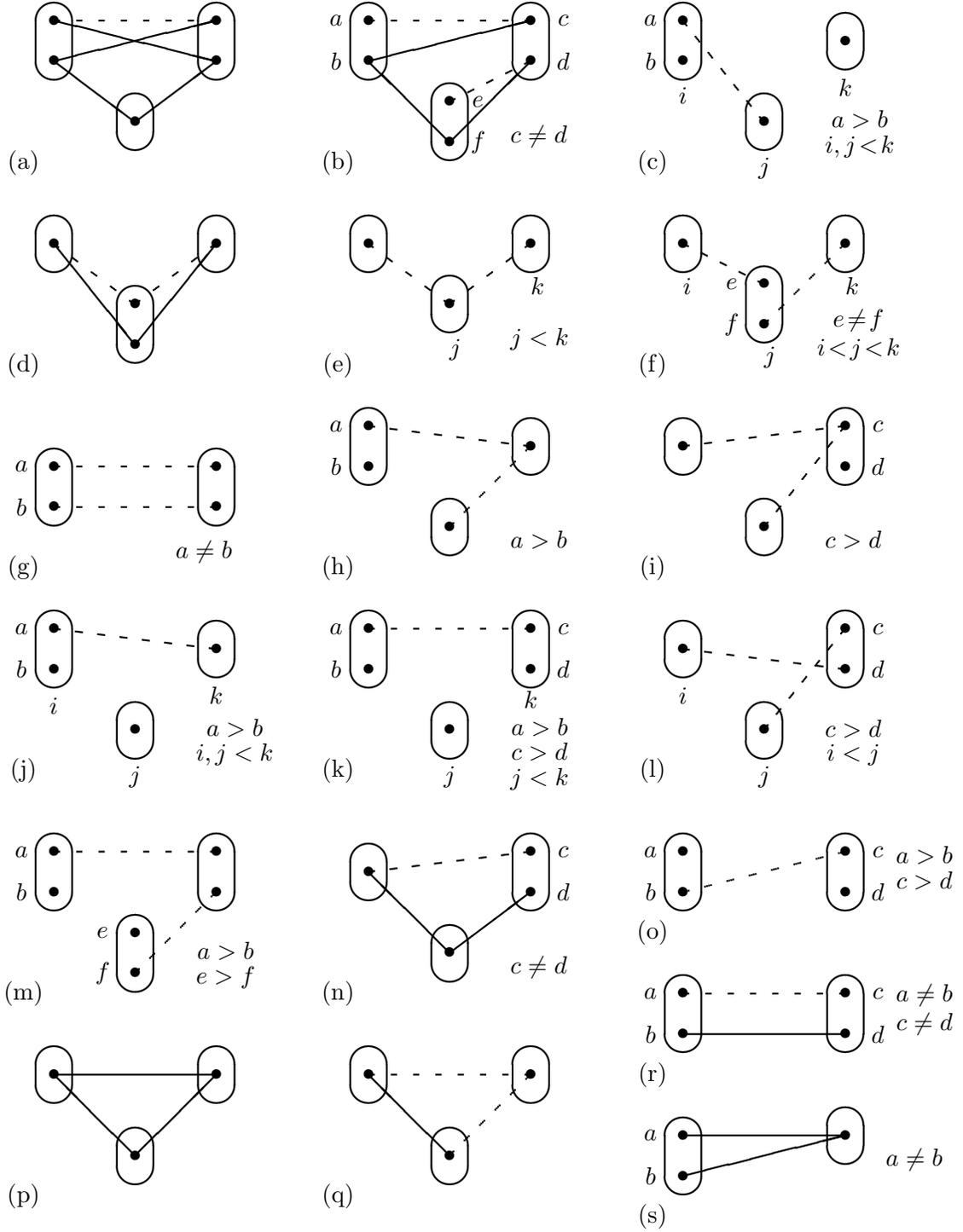
\begin{figure}       
\centering
\begin{picture}(480,620)(0,-20)

\put(0,500){
\begin{picture}(160,80)(0,0)
\put(10,50){\usebox{\varitwo}} \put(50,10){\usebox{\varione}}
\put(90,50){\usebox{\varitwo}}
\dashline{4}(20,80)(100,80)
\put(20,80){\line(4,-1){80}} \put(20,60){\line(4,1){80}}
\put(20,60){\line(4,-3){40}} \put(60,30){\line(4,3){40}}
\put(5,10){\makebox(0,0){(a)}}
\end{picture}}

\put(155,500){
\begin{picture}(160,90)(0,0)
\put(10,50){\usebox{\varitwo}} \put(50,10){\usebox{\varitwo}}
\put(90,50){\usebox{\varitwo}} \dashline{4}(20,80)(100,80)
\put(20,60){\line(4,1){80}} \put(60,20){\line(1,1){40}}
\put(60,20){\line(-1,1){40}} \dashline{4}(60,40)(100,60)
\put(4,80){\makebox(0,0){$a$}} \put(4,60){\makebox(0,0){$b$}}
\put(116,80){\makebox(0,0){$c$}} \put(116,60){\makebox(0,0){$d$}}
\put(74,40){\makebox(0,0){$e$}} \put(74,20){\makebox(0,0){$f$}}
\put(90,20){\makebox{$c \neq d$}}
\put(5,10){\makebox(0,0){(b)}}
\end{picture}}

\put(310,500){
\begin{picture}(150,90)(0,0)
\put(10,50){\usebox{\varitwo}} \put(50,10){\usebox{\varione}}
\put(90,50){\usebox{\varione}} \dashline{4}(20,80)(60,30)   
\put(4,80){\makebox(0,0){$a$}}
\put(4,60){\makebox(0,0){$b$}} \put(20,43){\makebox(0,0){$i$}}
\put(100,47){\makebox(0,0){$k$}} \put(60,7){\makebox(0,0){$j$}}
\put(90,15){\makebox{\shortstack{$a>b$ \\ $i,j \!<\! k$}}}   
\put(5,10){\makebox(0,0){(c)}}
\end{picture}}

\put(0,400){
\begin{picture}(160,80)(0,0)
\put(10,50){\usebox{\varione}} \put(50,10){\usebox{\varitwo}}
\put(90,50){\usebox{\varione}} \dashline{4}(20,70)(60,40)
\dashline{4}(60,40)(100,70) \put(20,70){\line(4,-5){40}}
\put(60,20){\line(4,5){40}} \put(5,10){\makebox(0,0){(d)}}
\end{picture}}

\put(155,400){
\begin{picture}(160,80)(0,0)
\put(10,50){\usebox{\varione}} \put(50,20){\usebox{\varione}}
\put(90,50){\usebox{\varione}} \dashline{4}(20,70)(60,40)
\put(60,15){\makebox{$j$}} \put(100,45){\makebox{$k$}}
\dashline{4}(60,40)(100,70) \put(90,20){\makebox{$j<k$}}
\put(5,10){\makebox(0,0){(e)}}
\end{picture}}

\put(310,400){
\begin{picture}(150,80)(0,0)
\put(10,50){\usebox{\varione}} \put(50,20){\usebox{\varitwo}}
\put(90,50){\usebox{\varione}} \dashline{4}(20,70)(60,50)
\put(20,45){\makebox{$i$}} \put(60,11){\makebox{$j$}}
\put(100,45){\makebox{$k$}} \dashline{4}(60,30)(100,70)
\put(44,50){\makebox(0,0){$e$}} \put(44,30){\makebox(0,0){$f$}}
\put(86,15){\makebox{\shortstack{$e\!\neq\!f$ \\ $i\!<\!j\!<\!k$}}}
\put(5,10){\makebox(0,0){(f)}}
\end{picture}}

\put(0,300){
\begin{picture}(160,70)(0,0)
\put(10,30){\usebox{\varitwo}} \put(90,30){\usebox{\varitwo}}
\dashline{4}(20,40)(100,40) \dashline{4}(20,60)(100,60)
\put(4,60){\makebox(0,0){$a$}} \put(4,40){\makebox(0,0){$b$}}
\put(80,15){\makebox{$a \neq b$}} \put(5,10){\makebox(0,0){(g)}}
\end{picture}}

\put(155,300){
\begin{picture}(160,80)(0,0)
\put(10,50){\usebox{\varitwo}} \put(50,10){\usebox{\varione}}
\put(90,50){\usebox{\varione}} \dashline{4}(20,80)(100,70)
\dashline{4}(60,30)(100,70) \put(4,80){\makebox(0,0){$a$}}
\put(4,60){\makebox(0,0){$b$}} \put(90,20){\makebox{$a>b$}}
\put(5,10){\makebox(0,0){(h)}}
\end{picture}}

\put(310,300){
\begin{picture}(160,80)(0,0)
\put(10,50){\usebox{\varione}} \put(50,10){\usebox{\varione}}
\put(90,50){\usebox{\varitwo}} \dashline{4}(20,70)(100,80)
\dashline{4}(60,30)(100,80) \put(116,80){\makebox(0,0){$c$}}
\put(116,60){\makebox(0,0){$d$}} \put(90,20){\makebox{$c>d$}}
\put(5,10){\makebox(0,0){(i)}}
\end{picture}}

\put(0,200){
\begin{picture}(150,80)(0,0)
\put(10,50){\usebox{\varitwo}} \put(50,10){\usebox{\varione}}
\put(90,50){\usebox{\varione}} \dashline{4}(20,80)(100,70)
\put(4,80){\makebox(0,0){$a$}} \put(4,60){\makebox(0,0){$b$}}
\put(20,42){\makebox(0,0){$i$}} \put(60,7){\makebox(0,0){$j$}}
\put(100,47){\makebox(0,0){$k$}}
\put(90,15){\makebox{\shortstack{$a>b$ \\ $i,j<k$}}}
\put(5,10){\makebox(0,0){(j)}}
\end{picture}}

\put(155,200){
\begin{picture}(150,80)(0,0)
\put(10,50){\usebox{\varitwo}} \put(50,10){\usebox{\varione}}
\put(90,50){\usebox{\varitwo}} \dashline{4}(20,80)(100,80)
\put(4,80){\makebox(0,0){$a$}} \put(4,60){\makebox(0,0){$b$}}
\put(116,80){\makebox(0,0){$c$}} \put(116,60){\makebox(0,0){$d$}}
\put(100,45){\makebox(0,0){$k$}} \put(60,7){\makebox(0,0){$j$}}
\put(90,3){\makebox{\shortstack{$a>b$ \\ $c>d$ \\ $j<k$}}}
\put(5,10){\makebox(0,0){(k)}}
\end{picture}}

\put(310,200){
\begin{picture}(150,80)(0,0)
\put(10,50){\usebox{\varione}} \put(50,10){\usebox{\varione}}
\put(90,50){\usebox{\varitwo}} \dashline{4}(20,70)(100,60)
\dashline{4}(60,30)(100,80) \put(116,80){\makebox(0,0){$c$}}
\put(116,60){\makebox(0,0){$d$}} \put(20,47){\makebox(0,0){$i$}}
\put(60,7){\makebox(0,0){$j$}}
\put(90,15){\makebox{\shortstack{$c>d$ \\
$i<j$}}} \put(5,10){\makebox(0,0){(l)}}
\end{picture}}

\put(0,90){
\begin{picture}(150,80)(0,0)
\put(10,50){\usebox{\varitwo}} \put(50,10){\usebox{\varitwo}}
\put(90,50){\usebox{\varitwo}} \dashline{4}(20,80)(100,80)
\dashline{4}(60,20)(100,60) \put(4,80){\makebox(0,0){$a$}}
\put(4,60){\makebox(0,0){$b$}} \put(44,40){\makebox(0,0){$e$}}
\put(44,20){\makebox(0,0){$f$}} \put(90,15){\makebox{\shortstack{$a>b$ \\
$e>f$}}} \put(5,10){\makebox(0,0){(m)}}
\end{picture}}

\put(155,90){
\begin{picture}(150,100)(0,0)
\put(10,50){\usebox{\varione}} \put(50,10){\usebox{\varione}}
\put(90,50){\usebox{\varitwo}} \dashline{4}(20,70)(100,80)
\put(60,30){\line(-1,1){40}} \put(60,30){\line(4,3){40}}
\put(116,80){\makebox(0,0){$c$}} \put(116,60){\makebox(0,0){$d$}}
\put(90,20){\makebox{$c \neq d$}} \put(5,10){\makebox(0,0){(n)}}
\end{picture}}

\put(310,120){
\begin{picture}(150,60)(0,-10)
\put(10,10){\usebox{\varitwo}} \put(90,10){\usebox{\varitwo}}
\put(4,40){\makebox(0,0){$a$}} \put(4,20){\makebox(0,0){$b$}}
\put(116,40){\makebox(0,0){$c$}} \put(116,20){\makebox(0,0){$d$}}
\dashline{4}(20,20)(100,40) \put(125,22){\makebox{\shortstack{$a>b$ \\
$c>d$}}} \put(5,0){\makebox(0,0){(o)}}
\end{picture}}

\put(0,-10){
\begin{picture}(150,80)(0,0)
\put(10,50){\usebox{\varione}} \put(50,10){\usebox{\varione}}
\put(90,50){\usebox{\varione}} \put(60,30){\line(-1,1){40}}
\put(20,70){\line(1,0){80}} \put(60,30){\line(1,1){40}}
\put(5,10){\makebox(0,0){(p)}}
\end{picture}}

\put(155,-10){
\begin{picture}(150,80)(0,0)
\put(10,50){\usebox{\varione}} \put(50,10){\usebox{\varione}}
\put(90,50){\usebox{\varione}} \put(60,30){\line(-1,1){40}}
\dashline{4}(20,70)(100,70) \dashline{4}(60,30)(100,70)
\put(5,10){\makebox(0,0){(q)}}
\end{picture}}

\put(310,50){
\begin{picture}(150,60)(0,-10)
\put(10,10){\usebox{\varitwo}} \put(90,10){\usebox{\varitwo}}
\put(20,20){\line(1,0){80}} \dashline{4}(20,40)(100,40)
\put(4,40){\makebox(0,0){$a$}} \put(4,20){\makebox(0,0){$b$}}
\put(116,40){\makebox(0,0){$c$}} \put(116,20){\makebox(0,0){$d$}}
\put(125,22){\makebox{\shortstack{$a \neq b$ \\
$c \neq d$}}} \put(5,0){\makebox(0,0){(r)}}
\end{picture}}

\put(310,-20){
\begin{picture}(150,60)(0,-10)
\put(10,10){\usebox{\varitwo}} \put(90,20){\usebox{\varione}}
\put(20,20){\line(4,1){80}} \put(20,40){\line(1,0){80}}
\put(4,40){\makebox(0,0){$a$}} \put(4,20){\makebox(0,0){$b$}}
\put(120,25){\makebox{$a \neq b$}} \put(5,0){\makebox(0,0){(s)}}
\end{picture}
}

\end{picture}
\caption{Patterns which do not occur in (a) $I_{K4}$; (b) $I_{4}$; (c) $I^{SAT}_{2\Delta}$; (d),(e),(f) $I_{5}$;
(g),(h),(i) $I^{SAT}_{6}$; (j),(k),(l),(m) $I^{SAT}_{K4}$; (n),(o),(p),(q),(r),(s) $I^{2COL}_{3}$.}
\label{fig:badpatterns}
\end{figure}    
\setlength{\unitlength}{1pt}

In figures representing CSP instances, similarly to patterns, ovals
represent variables, the set of points inside an oval the elements of
the variable's domain, a dashed (respectively, solid) line joining
two points represents the incompatibility (respectively,
compatibility) of the two points. In order not to clutter up figures
representing instances, trivial constraints containing only positive
edges  (i.e. corresponding to complete relations) are not shown.
Figure~\ref{fig:badpatterns}(a) 
is a pattern which does not occur in
the instance $I_{K4}$ (Figure~\ref{fig:IK4} on page~\pageref{fig:IK4}).  Similarly,
Figure~\ref{fig:badpatterns}(b) 
is a pattern which does not occur in
the instance $I_{4}$ (Figure~\ref{fig:I4} on page~\pageref{fig:I4}), and the pattern in
Figure~\ref{fig:badpatterns}(c) 
does not occur in instance
$I^{SAT}_{2\Delta}$. Figure~\ref{fig:badpatterns}(d), (e) and (f) 
are three patterns which do not occur in the instance $I_5$
(Figure~\ref{fig:I5} on page~\pageref{fig:I5}). The pattern (known as $T1$) shown in
Figure~\ref{fig:badpatterns}(d) 
is, in fact, a tractable pattern~\cite{Cooper15:dam}, but the fact that it does not occur in
$I_5$ (an arc-consistent instance which has no solution) shows that
arc consistency is not a decision procedure for
CSP$_{\overline{SP}}(T1)$. This instance was constructed using
certain known properties of the pattern
$T1$~\cite{Escamocher14:thesis}.

It can easily be verified that the three patterns
Figure~\ref{fig:badpatterns}(g), (h), (i) 
do not occur in $I^{SAT}_{6}$. Similarly,
the four patterns in Figure~\ref{fig:badpatterns}(j),(k),(l),(m) 
do not occur in the instance $I^{SAT}_{K4}$ (Figure~\ref{fig:ISATK4} on page~\pageref{fig:ISATK4}).

The instance $I^{2COL}_3$ is the problem of colouring a complete
graph on three vertices with only two colours. It is arc consistent
but clearly has no solution. It is easy to verify that none of the six
patterns in Figure~\ref{fig:badpatterns}(n),(o),(p),(q),(r),(s) 
occur in $I^{2COL}_3$.
Furthermore, trivially, no pattern on four or more variables occurs
in $I^{2COL}_3$ and no pattern with three or more distinct values in
the same domain occurs in $I^{2COL}_3$.

By Lemma~\ref{lem:not-ac}, we know that 
if there is one of 
the instances $I_{K4}$, $I_4$, $I^{SAT}_{2\Delta}$,
$I_{5}$, $I^{SAT}_6$, $I^{SAT}_{K4}$, $I^{2COL}_{3}$,
such that pattern $P$ does not occur in this instance
then $P$ is not AC-solvable. Let $P$ be any of the patterns
shown in Figure~\ref{fig:badpatterns}. 
By Lemma~\ref{lem:hered}, any pattern $Q$ in which $P$
occurs is not AC-solvable.

By the pattern in Figure~\ref{fig:badpatterns}(g), 
a simple AC-solvable pattern cannot contain two negative edges between the
same pair of variables. Since instance $I^{2COL}_3$ contains only
three variables and instance $I_5$ contains no triple of variables
which have a negative edge between each pair of variables, an
AC-solvable pattern can contain \emph{at most three variables and at
most two negative edges}. Thus to identify simple AC-solvable
patterns we only need to consider patterns on at most three
variables, at most two points per variable and with none, one or two
negative edges. Furthermore, in the case of two negative edges these
negative edges cannot be between the same pair of variables.

\subsection{Characterising AC-solvable unordered patterns}
\label{sec:unordered}

In this subsection, we consider only patterns $P$ that have no
associated structure (i.e. with $<_X \ =  \ <_D \ = \emptyset$).
We prove the following characterisation of unstructured AC-solvable patterns.

\begin{thm} \label{thm:noOrder}
If $P$ is a simple unordered pattern, then $P$ is AC-solvable if and
only if $P$ occurs in the pattern LX (Figure~\ref{fig:lx} on
page~\pageref{fig:lx}) or in the
pattern unordered(BTP).
\end{thm}

\begin{proof}
By the discussion in Section~\ref{sec:5.1}, we only need to consider
patterns $P$ with at most three variables, at most two points per
variable and at most two negative edges (with these edges not being
between the same pair of variables). We consider separately the cases
of a pattern with 0, 1 or 2 negative edges.

The only simple unordered pattern with no negative edges is the
triangle of positive edges shown in Figure~\ref{fig:badpatterns}(p) on
page~\pageref{fig:badpatterns} 
and this pattern is not AC-solvable since it does not occur in
$I^{2COL}_3$.

Let $P$ be a simple pattern with one negative edge $(a,c)$ (between
variables $y$ and $z$) and at most two points per variable. If the
domain of $y$ (respectively, $z$) contains another point $b$
(respectively, $d$), then for $a,b$ (respectively, $c,d$) to be
non-mergeable, there must be a positive edge $(b,c)$ (respectively,
$(a,d)$). Furthermore, for $b$ (respectively, $d$) not to be a
dangling point, it must belong to another positive edge. Any two
distinct points in the domain of a third variable $x$ would be
mergeable, so we can assume that $P$ has at most one point in the
domain of $x$. Since this point is not a dangling point, it must be
connected by positive edges to at least two points. By a simple
exhaustive search we can easily deduce that \emph{either} $P$ is (a
subpattern of) a triangle on three variables composed of one negative
edge and two positive edges (in which case $P$ occurs in the pattern
LX shown in Figure~\ref{fig:lx} on page~\pageref{fig:lx}), \emph{or} one of the patterns shown
in Figure~\ref{fig:badpatterns}(a), Figure~\ref{fig:badpatterns}(p)
or Figure~\ref{fig:badpatterns}(s) on page~\pageref{fig:badpatterns} occurs in $P$, in which case, by
Lemma~\ref{lem:hered}, $P$ is not AC-solvable.

Let $P$ be a simple pattern containing exactly two negative edges
$(a,b)$ (between variables $x,y$) and $(a,c)$ (between variables
$x,z$) that meet at the point $a$ of variable $x$. Suppose first that
$x$ has no other point. If $P$ does not occur in LX, then $P$ must
have a positive edge between variables $y$ and $z$ which is either
$(b,c)$ or $(d,e)$ where $d \neq b$ and $e \neq c$. In the latter
case, to avoid points $b,d$ (respectively, $c,e$) being mergeable,
$P$ must have the positive edge $(a,d)$ (respectively, $(a,e)$). We
can deduce that, if $P$ does not occur in LX, then one of the
patterns Figure~\ref{fig:badpatterns}(p) or
Figure~\ref{fig:badpatterns}(q) on page~\pageref{fig:badpatterns} occurs in $P$. Suppose now that $P$
has two points $a,f$ in the domain of variable $x$. For $a,f$ not to
be mergeable, $P$ must have either the positive edge $(b,f)$ or the
positive edge $(c,f)$. If it has both, then the pattern
Figure~\ref{fig:badpatterns}(d) on page~\pageref{fig:badpatterns} occurs in $P$. If $P$ has just one,
which without loss of generality we can suppose is the positive edge
$(b,f)$, then for $f$ not to be a dangling point, $f$ must belong to
another positive edge $(d,f)$ (where $d \neq b$) or $(e,f)$ (where $e
\neq c$). In the latter case, for $c,f$ not to be mergeable, $P$ must
also have the positive edge $(a,e)$. In both cases, the pattern
Figure~\ref{fig:badpatterns}(s) on page~\pageref{fig:badpatterns} occurs in $P$. Thus, if $P$ does not
occur in LX, then, by Lemma~\ref{lem:hered}, $P$ is not AC-solvable.

Finally, let $P$ be a simple pattern containing exactly two negative
edges $(a,b)$ (between variables $x,y$) and $(c,d)$ (between
variables $x,z$) with two distinct points $a \neq c$ in the domain of
variable $x$. To avoid $a,c$ being mergeable, $P$ must have a
positive edge $(b,c)$ or $(a,d)$. Without loss of generality, suppose
$P$ has the positive edge $(b,c)$. If $P$ does not occur in
unordered(BTP), then at least one of the variables $y,z$ must have
two distinct points. If $y$ has a point $e \neq b$, then for $b,e$
not to be mergeable, $P$ must have the positive edge $(a,e)$.
Similarly, if $z$ has a point $f \neq d$, then $P$ must have the
positive edge $(c,f)$. But then, to avoid dangling points, we have to
add other positive edges to $P$ and we find that one of the patterns
Figure~\ref{fig:badpatterns}(a), Figure~\ref{fig:badpatterns}(n),
Figure~\ref{fig:badpatterns}(p) or Figure~\ref{fig:badpatterns}(s)
on page~\pageref{fig:badpatterns}
occurs in $P$, and so, by Lemma~\ref{lem:hered}, $P$ is not
AC-solvable. By Lemma~\ref{lem:occ-unordered}, if $P$ occurs in 
unordered(BTP), then $P$ occurs in BTP, and thus is
AC-solvable~\cite{cjs10:aij-btp}.
\end{proof}

\subsection{Characterising AC-solvable variable-ordered patterns}

In this subsection we consider simple patterns $P$ which have no
domain order, (i.e. $<_D = \emptyset$), but do have a partial order
on the variables. We first require the following lemma.

\begin{lem} \label{lem:P<}
If $P^{<}$ is a pattern whose only structure is a partial order on
its variables and $P^{-}=\mbox{unordered}(P^{<})$, then
\begin{enumerate}
\item $P^{<}$ is simple if and only if $P^{-}$ is simple.
\item $P^{<}$ is AC-solvable only if $P^{-}$ is AC-solvable.
\end{enumerate}
\end{lem}

\begin{proof}
The property of being simple is (Definition~\ref{def:simple})
independent of any variable order, hence $P^{<}$ is simple if and
only if $P^{-}$ is simple. By Lemma~\ref{lem:occ-unordered}, $P^{-}$
occurs in $P^<$. The fact that $P^{<}$ is AC-solvable only if $P^{-}$
is AC-solvable then follows from Lemma~\ref{lem:hered}.
\end{proof}

Recall pattern LX$^{<}$ from Example~\ref{ex:lx<} that is obtained from the
pattern LX (Figure~\ref{fig:lx} on page~\pageref{fig:lx}) by adding the partial variable order $y<z$.
Recall (from from Example~\ref{ex:lx<}) that the patterns LX and  LX$^{<}$ are, in fact, equivalent.

Lemma~\ref{lem:P<} allows us to give the following characterisation of
variable-ordered AC-solvable patterns.

\begin{thm} \label{thm:varOrder}
If $P$ is a simple pattern whose only structure is a partial order on
its variables, then $P$ is AC-solvable if and only if $P$ occurs in
the pattern LX$^{<}$ (Example~\ref{ex:lx<}), the pattern BTP$^{vo}$
(Figure~\ref{fig:btp} on page~\pageref{fig:btp}) or the pattern invVar(BTP$^{vo}$).
\end{thm}

\begin{proof}
By Lemma~\ref{lem:P<} and Theorem~\ref{thm:noOrder}, we only need to
consider patterns $P^{-}$ occurring in LX or unordered(BTP) to which
we add a partial order on the variables to produce a pattern $P$.

We first consider the case of a pattern $P$ in which there are two negative
edges that meet. By Lemma~\ref{lem:P<} and Theorem~\ref{thm:noOrder},
for $P$ to be AC-solvable, $P^{-}$ must occur in LX. By
Theorem~\ref{thm:lx} and Example~\ref{ex:lx<}, we have that LX$^<$ is
AC-solvable. Let $L$ be the pattern composed of three variables
$x,y,z$ and two negative edges which meet at a point
$\tuple{x,\epsilon}$ (i.e. $L$ is the pattern LX without its positive
edges). Then Figure~\ref{fig:badpatterns}(e) on page~\pageref{fig:badpatterns} and Lemma~\ref{lem:inv}
tell us that placing any order between $x$ and $y$ or between $x$ and
$z$ turns $L$ into a pattern which is not AC-solvable. It follows
from Lemma~\ref{lem:hered} that an ordered pattern $P$ containing two
negative edges that meet is AC-solvable if and only if $P$ occurs in
$LX^{<}$.

Now consider simple patterns $P$ which contain two negative edges
that do not meet. By Lemma~\ref{lem:P<} and
Theorem~\ref{thm:noOrder}, for $P$ to be AC-solvable, $P^{-}$ must
occur in unordered(BTP). Adding almost any partial variable order to
unordered(BTP) produces a pattern which occurs in BTP$^{vo}$
(Figure~\ref{fig:btp}(b) on page~\pageref{fig:btp})
or invVar(BTP$^{vo}$). The only order for which this is not the case,
is the total order $x < z < y$ (or its inverse), where the variables
$x,y,z$ are as shown in Figure~\ref{fig:btp}(b) on page~\pageref{fig:btp}. Let $P$ be any
simple pattern on three variables $x,y,z$, containing two negative
edges (between $x,z$ and $y,z$) that do not meet and with the
variable order $x < z < y$. Then the pattern shown in
Figure~\ref{fig:badpatterns}(f) on page~\pageref{fig:badpatterns} occurs in $P$ and hence, by
Lemma~\ref{lem:hered}, $P$ is not AC-solvable.

If $P$ is any simple pattern such that $P^{-}$ occurs in LX or
unordered(BTP) and contains at most one negative edge, then $P$
occurs in the three-variable triangle pattern composed of one
negative and two positive edges; it is then easy to check that,
whatever the ordering of is variables, $P$ occurs in BTP$^{vo}$ or
invVar(BTP$^{vo}$).
\end{proof}

\subsection{Characterising AC-solvable domain-ordered patterns}

In this subsection we consider simple patterns $P$ with a partial
order on domains but no ordering on the variables.

Let EMC$^{-}$ be the no-variable-order version of the pattern EMC
depicted in Figure~\ref{fig:emc} on page~\pageref{fig:emc}.
We prove the following characterisation of domain-ordered AC-solvable patterns.

\begin{thm} \label{thm:domOrder}
If $P$ is a simple pattern whose only structure is a partial order on
its domains, then $P$ is AC-solvable if and only if $P$ occurs in the
pattern LX (Figure~\ref{fig:lx} on page~\pageref{fig:lx}), or the pattern EMC$^{-}$, or the
pattern invDom(EMC$^{-}$).
\end{thm}

\begin{proof}
As in the proofs of Theorem~\ref{thm:noOrder} and~\ref{thm:varOrder},
we only need to consider patterns on at most three variables, with at
most two points per variable and with either no negative edges, one
negative edge, two negative edges that meet or two negative edges
that do not meet. However, we have more cases to consider than in
Theorem~\ref{thm:noOrder} since patterns may now contain points, such
as $\beta$ in the pattern EMC shown in Figure~\ref{fig:emc} on
page~\pageref{fig:emc}, which
would be a dangling point without the domain order $\alpha > \beta$.

The only pattern with no negative edges and no mergeable points is
the triangle of positive edges shown in
Figure~\ref{fig:badpatterns}(p) on page~\pageref{fig:badpatterns} which is not AC-solvable.

Let $P$ be a simple pattern with one negative edge, at most two
points per variable and just two variables. If neither of the
patterns shown in Figure~\ref{fig:badpatterns}(o) on
page~\pageref{fig:badpatterns} and
Figure~\ref{fig:badpatterns}(r) on page~\pageref{fig:badpatterns} occur in $P$, then $P$ occurs in the
pattern EMC$^{-}$ or in invDom(EMC$^{-}$).

Let $P$ be a simple pattern on three variables $x,y,z$, with one
negative edge $(a,c)$ (between variables $y$ and $z$) and at most two
points per variable. If the domain of $y$ (respectively, $z$)
contains another point $b$ (respectively, $d$), then for $a,b$
(respectively, $c,d$) to be non-mergeable, there must be a positive
edge $(b,c)$ (respectively, $(a,d)$). Any two distinct points in the
domain of variable $x$ would be mergeable, so we can assume that $P$
has exactly one point in the domain of $x$. Since this point $e$ is
not a dangling point, it must be connected by positive edges to at
least two points. If none of the patterns in
Figure~\ref{fig:badpatterns}(a), Figure~\ref{fig:badpatterns}(p) and
Figure~\ref{fig:badpatterns}(s) on page~\pageref{fig:badpatterns} occurs in $P$, then $e$ must belong
to the two positive edges $(a,e)$ and $(c,e)$, and no others. If
neither of the patterns in Figure~\ref{fig:badpatterns}(o) or
Figure~\ref{fig:badpatterns}(r) on page~\pageref{fig:badpatterns} occurs in $P$, then we can deduce
that $P$ occurs in the pattern EMC$^{-}$ or in invDom(EMC$^{-}$).

Let $P$ be a simple pattern on three variables with at most two
points per variable and with two negative edges that meet. If $P$ has
two points $a,b$ in the domain of the same variable together with an
ordering $a<b$, then one of the patterns in
Figure~\ref{fig:badpatterns}(h) and Figure~\ref{fig:badpatterns}(i)
on page~\pageref{fig:badpatterns}
(or their domain-inversed version) occurs in $P$, and hence $P$
cannot be AC-solvable. This leaves only the case of unordered
patterns $P$. By the proof (and in particular the part that deals
with patterns containing exactly two negative edges that meet) of
Theorem~\ref{thm:noOrder}, we can deduce that if $P$ is AC-solvable
then it occurs in the pattern LX.

Let $P$ be a simple pattern on three variables $x,y,z$ with at most
two points per variable and with two negative edges $(a,c)$ (between
variables $x$ and $y$) and $(d,f)$ (between variables $y$ and $z$)
that do not meet (i.e. $c \neq d$). If $P$ contains only these four
points $a,c,d,f$, then it necessarily occurs in EMC$^{-}$. If $P$
contains exactly five points, then without loss of generality, we can
assume that there is a point $b \neq a$ in the domain of $x$. Since
$a,b$ are not mergeable, there must be a positive edge $(b,c)$ in
$P$. If $P$ has a positive edge $(b,f)$, then the pattern in
Figure~\ref{fig:badpatterns}(n) on page~\pageref{fig:badpatterns} occurs in $P$; if $P$ has a positive
edge $(b,d)$ then the pattern in Figure~\ref{fig:badpatterns}(s)
on page~\pageref{fig:badpatterns}
occurs in $P$. Now, if the pattern in Figure~\ref{fig:badpatterns}(o)
on page~\pageref{fig:badpatterns}
does not occur in $P$, then whatever ordering is placed on $a,b$ and
$c,d$, $P$ occurs in the pattern EMC$^{-}$ or in invDom(EMC$^{-}$).
If $P$ contains exactly six points, then there must be points $b \neq
a$ in the domain of $x$ and $e \neq f$ in the domain of $z$. Since
both $a,b$ and $e,f$, are not mergeable, there must be positive edges
$(b,c)$ and $(e,d)$. If $P$ has a positive edge $(b,e)$, then the
pattern in Figure~\ref{fig:badpatterns}(b) on page~\pageref{fig:badpatterns} occurs in $P$; if $P$ has
a positive edge $(b,f)$ or $(a,e)$, then the pattern in
Figure~\ref{fig:badpatterns}(n) on page~\pageref{fig:badpatterns} occurs in $P$; if $P$ has a positive
edge $(b,d)$ or $(e,c)$, then the pattern in
Figure~\ref{fig:badpatterns}(s) on page~\pageref{fig:badpatterns} occurs in $P$. But then in all other
cases one of $b$ and $e$ is a dangling point unless $P$ has an order
on both $a,b$ and $e,f$. But this then implies that at least one of
the patterns in Figure~\ref{fig:badpatterns}(m) and
Figure~\ref{fig:badpatterns}(o) on page~\pageref{fig:badpatterns} occur in $P$. In all these cases, by
Lemma~\ref{lem:hered}, $P$ is not AC-solvable.
\end{proof}

\subsection{Characterising AC-solvable ordered patterns}

In this subsection we consider the most general case of simple
patterns $P$ which have a partial domain order and a partial variable
order.
We prove the following characterisation of  AC-solvable patterns with partial
orders on domains and variables.

\begin{thm} \label{thm:order}
If $P$ is a simple pattern 
with a partial order on its
domains and/or variables, then $P$ is AC-solvable if and only if $P$ occurs in
one of the patterns LX$^{<}$, EMC (Figure~\ref{fig:emc} on page~\pageref{fig:emc}),
BTP$^{vo}$, BTP$^{do}$ (Figure~\ref{fig:btp} on page~\pageref{fig:btp}), BTX
(Figure~\ref{fig:btx} on page~\pageref{fig:btx}) or BTI
(Figure~\ref{fig:bti} on page~\pageref{fig:bti}) (or versions of these patterns with inversed domain-order
and/or variable-order).
\end{thm}

\begin{proof}
Let $P^{-}$ be the same pattern as $P$ but without the partial order
on its variables. If $P$ is AC-solvable, then, by
Lemma~\ref{lem:occ-unordered}, $P^{-}$ occurs in $P$, and hence, by
Lemmas~\ref{lem:occ-sup} and~\ref{lem:not-ac}, $P^-$ is also
AC-solvable. Thus, by Theorem~\ref{thm:domOrder}, $P^-$ must occur in
either LX, EMC$^{-}$ or invDom(EMC$^{-}$). We consider the four
cases: no negative edges, one negative edge, two negative edges that
meet, two negative edges that do not meet in $P$. We have already
seen in the proofs of Theorems~\ref{thm:noOrder}
and~\ref{thm:domOrder} that there are no simple AC-solvable patterns
with only positive edges, so there remain three cases to consider.

\begin{figure}   
\centering
\begin{picture}(380,100)(0,0)

\put(0,0){\begin{picture}(180,100)(0,0) \put(10,60){\usebox{\varone}}
\put(50,10){\usebox{\varone}} \put(90,50){\usebox{\vartwo}}
\dashline{5}(20,80)(100,80) \put(20,80){\line(4,-1){80}}
\put(20,80){\line(4,-5){40}} \put(60,30){\line(4,5){40}}
\put(60,10){\makebox(0,0){$j$}} \put(20,57){\makebox(0,0){$i$}}
\put(100,42){\makebox(0,0){$k$}}
  \put(130,65){\makebox{$c>d$}}
  \put(5,80){\makebox(0,0){$a$}}
  \put(115,80){\makebox(0,0){$c$}}
  \put(115,60){\makebox(0,0){$d$}}  \put(15,5){\makebox(0,0){(a)}}
\end{picture}}

\put(200,0){\begin{picture}(180,100)(0,0)
\put(10,50){\usebox{\vartwo}} \put(50,10){\usebox{\varone}}
\put(90,50){\usebox{\vartwo}} \dashline{5}(20,80)(100,80)
\put(20,80){\line(4,-1){80}} \put(20,60){\line(4,1){80}}
\put(20,80){\line(4,-5){40}} \put(60,30){\line(4,5){40}}
\put(60,10){\makebox(0,0){$j$}} \put(20,42){\makebox(0,0){$i$}}
\put(100,42){\makebox(0,0){$k$}}
  \put(130,60){\makebox{\shortstack{$a>b$ \\ $c>d$}}}
  \put(5,80){\makebox(0,0){$a$}}
  \put(5,60){\makebox(0,0){$b$}}
  \put(115,80){\makebox(0,0){$c$}}
  \put(115,60){\makebox(0,0){$d$}}  \put(15,5){\makebox(0,0){(b)}}
\end{picture}}

\end{picture}
\caption{One-negative-edge patterns occurring in EMC$^{-}$.}
\label{fig:sub1emc}
\end{figure}
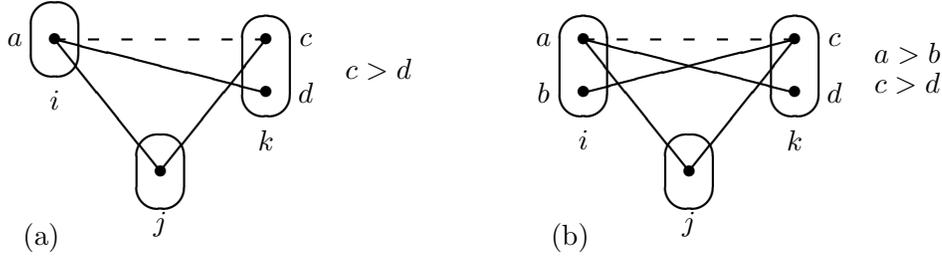  

\begin{figure}   
\centering
\begin{picture}(380,100)(0,0)

\put(0,0){\begin{picture}(180,100)(0,0) \put(10,60){\usebox{\varone}}
\put(50,10){\usebox{\varone}} \put(90,50){\usebox{\vartwo}}
\dashline{5}(20,80)(100,80)  \dashline{5}(60,30)(100,60)
\put(20,80){\line(4,-1){80}} \put(20,80){\line(4,-5){40}}
\put(60,30){\line(4,5){40}} \put(60,10){\makebox(0,0){$j$}}
\put(20,57){\makebox(0,0){$i$}}  \put(100,42){\makebox(0,0){$k$}}
  \put(130,65){\makebox{$c>d$}}
  \put(5,80){\makebox(0,0){$a$}}
  \put(115,80){\makebox(0,0){$c$}}
  \put(115,60){\makebox(0,0){$d$}}  \put(15,5){\makebox(0,0){(a)}}
\end{picture}}

\put(200,0){\begin{picture}(180,100)(0,0)
\put(10,50){\usebox{\vartwo}} \put(50,10){\usebox{\varone}}
\put(90,50){\usebox{\vartwo}} \dashline{5}(20,80)(100,80)
\dashline{5}(60,30)(100,60) \put(20,80){\line(4,-1){80}}
\put(20,60){\line(4,1){80}} \put(20,80){\line(4,-5){40}}
\put(60,30){\line(4,5){40}} \put(60,10){\makebox(0,0){$j$}}
\put(20,42){\makebox(0,0){$i$}}  \put(100,42){\makebox(0,0){$k$}}
  \put(130,60){\makebox{\shortstack{$a>b$ \\ $c>d$}}}
  \put(5,80){\makebox(0,0){$a$}}
  \put(5,60){\makebox(0,0){$b$}}
  \put(115,80){\makebox(0,0){$c$}}
  \put(115,60){\makebox(0,0){$d$}}  \put(15,5){\makebox(0,0){(b)}}
\end{picture}}

\end{picture}
\caption{Two-negative-edge patterns occurring in EMC.}
\label{fig:sub2emc}
\end{figure}
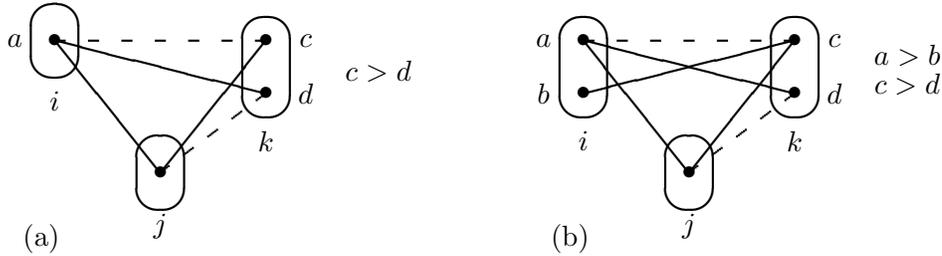  


Let $P$ be a simple pattern with exactly one negative edge. If $P^-$
(which also has one negative edge) occurs in LX and is simple (and
hence has no dangling points) then it must occur in a triangle $T$
consisting of one negative and two positive edges. This triangle
pattern $T$ occurs in BTP or in invVar(BTP) whatever ordering we
place on its three variables and hence the same is true of $P$. If
$P^-$ (which has one negative edge and is unmergeable) occurs in
EMC$^{-}$ but not in the triangle $T$, then $P^-$ occurs in one of
the two patterns shown in Figure~\ref{fig:sub1emc} on page~\pageref{fig:sub1emc} (or the
domain-inversed versions of these patterns) and includes the point
$d$ (in the case corresponding to Figure~\ref{fig:sub1emc}(a) on
page~\pageref{fig:sub1emc}) or the
points $b$ and $d$ (in the case corresponding to
Figure~\ref{fig:sub1emc}(b)) on page~\pageref{fig:sub1emc}. For $d$ (respectively, $b$) not to be a
dangling point in $P$, we must have the order $c>d$ (respectively
$a>b$) in $P$. If $P^-$ occurs in the pattern in
Figure~\ref{fig:sub1emc}(a) on page~\pageref{fig:sub1emc} and $P$ has the variable order $i,j<k$
(or $i<k$ or $j<k$), then $P$ occurs in BTP$^{do}$. If $P^-$ occurs
in the pattern in Figure~\ref{fig:sub1emc}(a) on page~\pageref{fig:sub1emc} and $P$ has the
variable order $i<j$, then $P$ occurs in invVar(BTI). If $P^-$ occurs
in the pattern in Figure~\ref{fig:sub1emc}(a) on page~\pageref{fig:sub1emc} and $P$ includes the
variable order $i<j,k$ and domain order $c>d$, then the pattern in
Figure~\ref{fig:badpatterns}(j) on page~\pageref{fig:badpatterns} occurs in invVar($P$). If $P^-$
occurs in the pattern in Figure~\ref{fig:sub1emc}(a) on
page~\pageref{fig:sub1emc} and $P$ includes
the variable order $i,k<j$ and domain order $c>d$, then the pattern
in Figure~\ref{fig:badpatterns}(c) on page~\pageref{fig:badpatterns} occurs in $P$. This covers all
variable orderings of $P$ (after taking into account the
variable-inversed versions of each case) when $P^-$ occurs in the
pattern in Figure~\ref{fig:sub1emc}(a) on page~\pageref{fig:sub1emc}. Now consider the case in
which $P^-$ occurs in the pattern in Figure~\ref{fig:sub1emc}(b) on
page~\pageref{fig:sub1emc}. If
$P$ includes the variable order $j<k$ (or $j<i$) together with the
domain order $a>b$ and $c>d$, then the pattern in
Figure~\ref{fig:badpatterns}(k) on page~\pageref{fig:badpatterns} occurs in $P$. If $P$ has the
variable order $i<k$ then $P$ occurs in EMC. Thus all
one-negative-edge AC-solvable patterns occur in BTP, BTI or EMC (or
their domain and/or variable-inversed versions).

Let $P$ be a simple pattern with two negative edges that meet at a
point. $P^-$ necessarily occurs in the pattern LX. Let $j$ be the
variable of $P$ where the two negative edges meet, and let $i,k$ be
the other two variables. If $P$ includes the variable order $j<k$
(or, by symmetry, the order $j<i$), then the pattern in
Figure~\ref{fig:badpatterns}(e) on page~\pageref{fig:badpatterns} occurs in $P$ and hence $P$ is not
AC-solvable. If $P$ has the variable order $i<k$, then $P$ occurs in
the pattern $LX^{<}$ (the version of LX shown in Figure~\ref{fig:lx}
on page~\pageref{fig:lx}
together with the variable order $y<z$). By symmetry, we have covered
all possible cases.

Finally, let $P$ be a simple pattern with two negative edges that do
not meet at a point. $P^-$ necessarily occurs in the pattern EMC. We
distinguish two distinct cases: (1) $P^-$ occurs in the pattern in
Figure~\ref{fig:sub2emc}(a) on page~\pageref{fig:sub2emc} or (2) $P^-$ occurs in the pattern
 in Figure~\ref{fig:sub2emc}(b) on page~\pageref{fig:sub2emc} and includes the point $b$ together with the order
$a>b$ (otherwise $b$ would be a dangling point). We first consider
case (1). If $P$ includes the order $c>d$ and $i<j$, then the
(inversed domain-order version of the) pattern in
Figure~\ref{fig:badpatterns}(l) on page~\pageref{fig:badpatterns} occurs in $P$ and hence $P$ is not
AC-solvable. If $P$ includes the order $i<k<j$, then the pattern in
Figure~\ref{fig:badpatterns}(f) on page~\pageref{fig:badpatterns} occurs in $P$. If $P$ has the
variable order $i<j<k$ and no domain order, then $P$ occurs in
BTP$^{vo}$. If $P$ has the variable order $i,j<k$ and the domain
order $c>d$, then $P$ occurs in BTP$^{do}$. All other patterns which
fall in case (1) are covered by symmetry. Now we consider the case
(2). First suppose that $P$ includes the domain order $c>d$ (as well
as $a>b$). If $P$ includes the variable order $i<j$, then the
(domain-inversed version of the) pattern in
Figure~\ref{fig:badpatterns}(l) on page~\pageref{fig:badpatterns} occurs in $P$. If $P$ includes the
variable order $j<k$, then the pattern in
Figure~\ref{fig:badpatterns}(k) on page~\pageref{fig:badpatterns} occurs in $P$. If $P$ has the
variable order $i<k$, then $P$ occurs in EMC. Now, suppose that $P$
does not include the domain order $c>d$. If $P$ includes the variable
order $i,j<k$, then the pattern in Figure~\ref{fig:badpatterns}(j)
on page~\pageref{fig:badpatterns}
occurs in $P$. If $P$ includes the variable order $i,k<j$, then the
pattern in Figure~\ref{fig:badpatterns}(c) on page~\pageref{fig:badpatterns} occurs in $P$. If $P$ has
the variable order $i<j,k$ (or $i<j$ or $i<k$) then $P$ occurs in
BTX. If $P$ has the variable order $j<k$ then $P$ occurs in BTI. By
symmetry we have covered all possible variable orderings of $P$ in
case (2).
\end{proof}

\section{Conclusion}

We have identified 4 new tractable classes of binary CSPs. Moreover,
we have given a characterisation of all simple partially-ordered
patterns decided by AC. We finish with open problems.

For future work, we plan to study the wider class of unmergeable
ordered patterns in which two points $a,b$ may be non-mergeable
simply because there is an order $a<b$ on them. In the present paper,
$a,b$ are mergeable unless they have different compatibilities with a
third point $c$.

Is there a way of giving a unified description of EMC, BTX and BTI, since to find a solution
after establishing arc consistency we use basically the same
algorithm? Any such generalisation will not be a simple forbidden
pattern by Theorem~\ref{thm:order}, but
there is possibly some other way of combining these patterns.

Are there interesting generalisations of these patterns to constraints of
arbitrary arity, valued constraints, infinite domains or QCSP? BTP has been
generalised to constraints of arbitrary arity~\cite{cooper14:cp-broken} as well
as to QCSPs~\cite{Gao11:aaai}. Max-closed constraints have been generalised to
VCSPs~\cite{Cohen06:complexitysoft}. Infinite domains is an interesting avenue
of future research because simple temporal constraints are binary max-closed~\cite{Dechter91:temporal}.

In this paper, we only focused on classes of CSP instances with
totally ordered domains (but defined by partially-ordered patterns).
However, the framework of forbidden patterns captures language-based
CSPs with partially-ordered domains, such as CSPs with a semi-lattice
polymorphism. In the future, we plan to investigate classes of CSP
instances with partially-ordered domains.

\paragraph{\bf Acknowledgement} We would like to thank the reviewers of both the
extended abstract of this paper~\cite{cz16:lics} and this full version for their very 
detailed comments and pertinent suggestions.



\bibliographystyle{plain}
\bibliography{ac}

\end{document}